\documentclass[final]{IEEEtran}
\usepackage[T1]{fontenc}
\usepackage[latin9]{inputenc}
\usepackage{array}
\usepackage{verbatim}
\usepackage{amsmath}
\usepackage{amssymb}
\usepackage{graphicx}
\usepackage{hyperref}
\usepackage{breakurl}

\usepackage{cite}\usepackage{amsthm}\usepackage{dsfont}\usepackage{array}\usepackage{mathrsfs}\usepackage{comment}
\usepackage{subfig}
\usepackage{epsfig}
\usepackage{epstopdf}
\usepackage[ruled]{algorithm2e}

\newcommand{\bSigma}{{\boldsymbol \Sigma}}
\newcommand{\bbeta}{{\boldsymbol \beta}}

\newcommand{\ba}{{\boldsymbol a}}
\newcommand{\cA}{{\mathcal A}}
\newcommand{\cS}{{\mathcal S}}

\newcommand{\bw}{{\boldsymbol w}}
\newcommand{\bU}{{\boldsymbol U}}
\newcommand{\bV}{{\boldsymbol V}}
\newcommand{\bv}{{\boldsymbol v}}

\newcommand{\bI}{{\boldsymbol I}}
\newcommand{\bY}{{\boldsymbol Y}}
\newcommand{\bH}{{\boldsymbol H}}

\newcommand{\bX}{{\boldsymbol X}}
\newcommand{\bx}{{\boldsymbol x}}

\newcommand{\bz}{{\boldsymbol z}}
\newcommand{\bZ}{{\boldsymbol Z}}
\newcommand{\bQ}{{\boldsymbol Q}}

\newcommand{\argmin}{\mathrm{argmin}}
\newcommand{\supp}{\mathrm{supp}}
\newcommand{\sgn}{\mathrm{sgn}}
\usepackage{cite}\usepackage{amsthm}\usepackage{dsfont}\usepackage{array}\usepackage{mathrsfs}\usepackage{comment}

\usepackage{hyperref}
\usepackage{breakurl}

\usepackage{color}
\usepackage[normalem]{ulem}
\newcommand\yc[1]{{\color{black}#1}}

\newcommand\yxl[1]{{\color{black}#1}}

\begin{document}

\theoremstyle{plain}\newtheorem{lemma}{\textbf{Lemma}}\newtheorem{theorem}{\textbf{Theorem}}\newtheorem{corollary}{\textbf{Corollary}}\newtheorem{assumption}{\textbf{Assumption}}\newtheorem{example}{\textbf{Example}}\newtheorem{definition}{\textbf{Definition}}

\theoremstyle{definition}

\theoremstyle{remark}\newtheorem{remark}{\textbf{Remark}}

\title{Low-Rank Positive Semidefinite Matrix Recovery from \yc{Corrupted Rank-One Measurements}}
 
\author{Yuanxin Li, \IEEEmembership{Student Member,~IEEE}, Yue Sun, and Yuejie Chi, \IEEEmembership{Member,~IEEE}$^{\star}$\thanks{Y. Li and Y. Chi are with Department of Electrical and Computer Engineering, The Ohio State University, Columbus, OH 43210 USA (e-mails: \{li.3822, chi.97\}@osu.edu). Y. Sun is with Department of Electronics Engineering, Tsinghua University, Beijing, China. Part of the work was done while Y. Sun was visiting The Ohio State University. }
\thanks{This work is supported in part by NSF under grant CCF-1422966, ECCS-1462191 and AFOSR under grant FA9550-15-1-0205. Corresponding e-mail: chi.97@osu.edu. Date: \today.}\thanks{Preliminary results of this paper were presented in part at the IEEE International Conference on Acoustics, Speech and Signal Processing, Shanghai, China, March 2016. } }

\maketitle

\begin{abstract}
We study the problem of estimating a low-rank positive semidefinite (PSD) matrix from a set of rank-one measurements using sensing vectors composed of i.i.d. standard Gaussian entries, which are possibly corrupted by arbitrary outliers. This problem arises from applications such as phase retrieval, covariance sketching, quantum space tomography, and power spectrum estimation. We first propose a convex optimization algorithm that seeks the PSD matrix with the minimum $\ell_1$-norm of the observation residual. The advantage of our algorithm is that it is free of parameters, therefore eliminating the need for tuning parameters and allowing easy implementations. We establish that with high probability, a low-rank PSD matrix can be exactly recovered as soon as the number of measurements is large enough, even when a fraction of the measurements are corrupted by outliers with arbitrary magnitudes. Moreover, the recovery is also stable against bounded noise. With the additional information of an upper bound of the rank of the PSD matrix, we propose another non-convex algorithm based on subgradient descent that demonstrates excellent empirical performance in terms of computational efficiency and accuracy.

\end{abstract}


\section{Introduction}
In many emerging applications of science and engineering, we are interested in estimating a low-rank positive semidefinite (PSD) matrix $\boldsymbol{X}_{0}\in\mathbb{R}^{n\times n}$ from a set of nonnegative magnitude measurements: 
\begin{equation}\label{general_model}
z_{i} = \langle \boldsymbol{Z}_{i} , \boldsymbol{X}_{0}\rangle =  \langle \boldsymbol{a}_{i}\boldsymbol{a}_{i}^{T}, \boldsymbol{X}_{0}\rangle= \boldsymbol{a}_{i}^{T}\boldsymbol{X}_{0}\boldsymbol{a}_{i},  
\end{equation}
for $i=1,\ldots, m$, where $\langle \cdot,\cdot \rangle$ denotes the inner product operator. The measurement $z_i$ is quadratic in the sensing vector $\ba_i\in\mathbb{R}^n$, but linear in $\bX_0$, where the sensing matrix $\bZ_i = \ba_i\ba_i^T$ is {\em rank-one}. On one hand, such magnitude measurements could arise due to physical limitations, e.g. incapability of capturing phases, such as in phase retrieval and optical imaging from intensity measurements \cite{fienup1978reconstruction, candes2013phaselift, candes2013phase, waldspurger2015phase,schniter2015compressive,shechtman2015phase}, where only the squared intensity of linear measurements of a signal $\bx_0\in\mathbb{R}^n$ is recorded:
\begin{equation}\label{phase_retrieval}
z_i =  \left|\langle \ba_i, \bx_{0}\rangle\right|^2 = \ba_i^T \left(\bx_0\bx_0^T\right) \ba_i = \ba_i^T \bX_0\ba_i,
\end{equation}
where $\bX_0=\boldsymbol{x}_{0}\boldsymbol{x}_{0}^{T}$ is a lifted rank-one matrix from the signal $\bx_0$ of interest. On the other hand, they could arise by design, such as from the covariance sketching scheme considered in \cite{chen2015exact}, where $z_i$ is aggregated from squared intensity measurements of $L$ data samples of a zero-mean ergodic data stream $\{\bx_l\}_{l=1}^{\infty}$ as
\begin{equation}\label{covariance_sketching}
z_i = \frac{1}{L} \sum_{l=1}^{L} \left|\langle \ba_i, \bx_{l}\rangle\right|^2 = \ba_i^T\left(\frac{1}{L} \sum_{l=1}^{L}\bx_{l}\bx_{l}^T \right) \ba_i \approx \ba_i^T\bX_0\ba_i.
\end{equation}
\yc{Here, $\bX_0=\mathbb{E}[\bx_l\bx_l^T]$ corresponds to the covariance matrix of the data when $L$ is sufficiently large, and the goal of covariance sketching is to recover the covariance matrix $\bX_0$ from the set of measurements $\{z_i\}_{i=1}^m$. In many applications such as array signal processing \cite{scharf1991statistical} and network traffic monitoring \cite{lakhina2004structural}, the covariance matrix of the data can be well approximated by a low-rank PSD matrix, as most of its variance can be explained by the few top principal components.} Last but not least, measurements of low-rank PSD matrices in the form of \eqref{general_model} also occur in a number of applications such as quantum state tomography \cite{gross2010quantum}, compressive power spectrum estimation \cite{ariananda2012compressive}, non-coherent direction-of-arrival estimation from magnitude measurements \cite{kim2015non}, synthetic aperture radar imaging \cite{mason2015passive}, and so on.



It is natural to ask if it is possible to recover the low-rank PSD matrix $\bX_0$ in \eqref{general_model} from an information-theoretically optimal number of measurements in a computationally efficient manner. A popular approach is based on convex relaxation \cite{chen2015exact}, which seeks the PSD matrix with the smallest trace norm while satisfying the observation constraint. It is shown in \cite{chen2015exact} that this algorithm exactly recovers all rank-$r$ PSD matrices as soon as the number of measurements exceeds the order of $nr$ in the absence of noise, and the recovery is stable against bounded noise as well.


\subsection{Our Goal and Contributions}
In this paper, we focus on robust recovery of the low-rank PSD matrix when the measurements in \eqref{general_model} are further corrupted by outliers, possibly adversarial with arbitrary amplitudes. In signal processing applications, outliers are somewhat inevitable, which may be caused by sensor failures, malicious attacks, or reading errors. In the application of covariance sketching, as in \eqref{covariance_sketching}, a sufficient aggregation length $L$ is necessary in order for each measurement $z_i$ to be well approximated by \eqref{general_model}. Measurements which are not aggregated from a large enough $L$ may be regarded as outliers. Therefore, it becomes critical to address robust recovery of $\bX_0$ in the presence of outliers. Fortunately, it is reasonable to assume that the number of outliers is usually much smaller than the number of total measurements, making it possible to leverage the sparsity of the outliers to faithfully recover the low-rank PSD matrix of interest.

We first propose a convex optimization algorithm that seeks the PSD matrix that minimizes the $\ell_{1}$-norm of the measurement residual, where the $\ell_1$-norm is adopted to promote outlier sparsity. The proposed convex program is free of tuning parameters and eliminates the need for trace minimization, a popular convex surrogate for low-rank matrix recovery, by only enforcing the PSD constraint. Neither does it require the knowledge of the outliers, even their existence. When the sensing vectors are composed of i.i.d. standard Gaussian entries, we establish that for a fixed $n\times n$ rank-$r$ PSD matrix, as long as the number of measurements exceeds the order of $nr^{2}$, the proposed convex program can exactly recover it with high probability, even when a fraction of an order of $1/r$ measurements are arbitrarily corrupted. Our measurement complexity is order-wisely near-optimal up to a factor of $r$, and is near-optimal in the rank-one case up to a constant factor. Furthermore, the recovery is also stable against additive bounded noise. While the proposed convex program coincides with a version of the PhaseLift algorithm \cite{candes2012solving,demanet2012stable,hand2016phaselift} studied in the literature for phase retrieval, our work provides its first theoretical performance guarantee to recover low-rank PSD matrices in the presence of arbitrary outliers. Moreover, we show the proposed approach can be easily extended to recover low-rank Toeplitz PSD matrices via numerical simulations.

To further reduce the computational burden when facing large-scale problems, we next develop a non-convex algorithm based on {\em subgradient descent} when the rank of the PSD matrix, or an upper bound of it, is known a priori. Since any rank-$r$ PSD matrix can be uniquely decomposed as $\bX_0=\bU_0\bU_0^T$, where $\bU_0\in\mathbb{R}^{n\times r}$ up to some orthonormal transformations, it is sufficient to recover $\bU_0$ without constructing the PSD matrix explicitly. The subgradient descent algorithm then iteratively updates the estimate by descending along the subgradient of the $\ell_1$-norm of the measurement residual using a properly selected step size and spectral initialization. We conduct extensive numerical experiments to demonstrate its excellent empirical performance, and compare it against the convex program proposed above as well as other alternative approaches in the literature.

\subsection{Organization}

The rest of the paper is organized as below. Section~\ref{sec:convex_algorithm} presents the proposed convex optimization algorithm and its corresponding performance guarantee, where detailed comparisons to related work are presented. Section~\ref{sec:non_convex} describes the proposed non-convex subgradient descent algorithm that is computationally efficient with excellent empirical performance. Numerical examples are provided in Section~\ref{sec:numerical}. The proof of the main theorem is given in Section~\ref{sec:proofs}. Finally, we conclude in Section~\ref{sec:conclusion}.

\section{Parameter-Free Convex Relaxation}\label{sec:convex_algorithm}
\subsection{Problem Formulation}
Let $\bX_0\in\mathbb{R}^{n\times n}$ be a rank-$r$ PSD matrix, then the set of $m$ measurements, which may be corrupted by either arbitrary outliers or bounded noise, can be represented as
\begin{equation}\label{measurement_operator}
\boldsymbol{z} = \mathcal{A}(\bX_0) + \bbeta + \bw,
\end{equation}
where $\boldsymbol{z},\boldsymbol{\beta},\boldsymbol{w}\in\mathbb{R}^{m}$. The linear mapping $\mathcal{A}$: $\mathbb{R}^{n\times n}\to\mathbb{R}^{m}$ is defined as $\mathcal{A}\left(\boldsymbol{X}_{0}\right)=\left\{\boldsymbol{a}_{i}^{T}\boldsymbol{X}_{0}\boldsymbol{a}_{i}\right\}_{i=1}^{m}$, \yc{where $\ba_i\in\mathbb{R}^n$ is the $i$th sensing vector composed of i.i.d. standard Gaussian entries, $i=1,\ldots, m$}. The vector $\bbeta$ denotes the outlier vector, which is assumed to be sparse whose entries can be arbitrarily large. The fraction of nonzero entries is defined as $s :=\left\Vert\boldsymbol{\beta}\right\Vert_{0}/m$. Moreover, the vector $\bw$ denotes the additive noise, which is assumed bounded as $\left\Vert\boldsymbol{w}\right\Vert_{1}\le\epsilon$. Our goal is to robustly recover $\bX_0$ from the measurements $\bz$.

\subsection{Recovery via Convex Relaxation}

To motivate our algorithm, consider the case when only the outlier vector $\bbeta$ is present in \eqref{measurement_operator} and the rank of $\bX_0$ is known. One may seek a rank-$r$ PSD matrix that minimizes the cardinality of the measurement residual to motivate outlier sparsity, given as
\begin{equation}\label{non_convex_outliers}
\hat{\bX} = \argmin_{\bX\succeq 0}\|\boldsymbol{z} - \mathcal{A}(\bX)\|_0, \quad \mbox{s.t.} \quad \mbox{rank}(\bX) = r.
\end{equation}
However, both the cardinality minimization and the rank constraint are NP-hard in general, making this method computationally infeasible. A common approach is to resort to convex relaxation, where we relax the cardinality minimization by its convex relaxation, i.e. the $\ell_1$-norm, and meanwhile, drop the rank constraint, yielding:
\begin{equation}\label{phaselift_outlier}
\yc{(\mbox{Robust-PhaseLift:})}\quad \hat{\bX} = \argmin_{\bX\succeq 0} \|\boldsymbol{z} - \mathcal{A}(\bX)\|_1 .
\end{equation}
\yc{We denote the above convex program as the Robust-PhaseLift algorithm, since it coincides with the PhaseLift algorithm studied in \cite{candes2012solving,demanet2012stable,hand2016phaselift} for phase retrieval\footnote{Note that there are a few different versions of PhaseLift in the literature which are not outlier-robust, therefore we rename \eqref{phaselift_outlier} to Robust-PhaseLift for emphasis.}.} The advantage of Robust-PhaseLift in \eqref{phaselift_outlier} is that it does not require any prior knowledge of the noise bound, the rank of $\bX_0$, nor the sparsity level of the outliers, and is free of any regularization parameter. \yc{It is also worth emphasizing that due to the special rank-one measurement operator, in \eqref{phaselift_outlier} it is possible to only honor the PSD constraint but not motivate the low-rank structure explicitly, via for example, trace minimization\footnote{The interested readers are invited to look up Fig. 1 in \cite{candes2013phaselift} for an intuitive geometric interpretation in the noise-free and outlier-free case. }.}

Encouragingly, we demonstrate that the algorithm~\eqref{phaselift_outlier} admits robust recovery of a rank-$r$ PSD matrix as soon as the number of measurements is large enough, even with a fraction of arbitrary outliers in Theorem~\ref{main}. To the best of our knowledge, this is the first theoretical performance guarantee of the robustness of \eqref{phaselift_outlier} with respect to arbitrary outliers in the low-rank setting. Our main theorem is given as below.

\begin{theorem}\label{main}
Suppose that $\|\bw\|_1\leq \epsilon$ and $s=\left\Vert\boldsymbol{\beta}\right\Vert_{0}/m$. Assume the support of $\bbeta$ is selected uniformly at random with the signs of its nonzero entries generated from the Rademacher distribution as $\mathbb{P}\left\{\mathrm{sgn}\left(\beta_{i}\right)=-1\right\} = \mathbb{P}\left\{\mathrm{sgn}\left(\beta_{i}\right)=1\right\} = 1/2$ for each $i \in \supp(\bbeta)$. Then for a fixed rank-$r$ PSD matrix $\boldsymbol{X}_0\in\mathbb{R}^{n\times n}$, there exist some absolute constants $c_{1}>0$ and $0<s_{0}<1$ such that as long as
\begin{equation*}
m \ge c_1nr^{2}, \quad s\leq \frac{s_0}{r},
\end{equation*} 
the solution to \eqref{phaselift_outlier} satisfies
$$ \left\|\hat{\boldsymbol{X}} - \boldsymbol{X}_0 \right\|_{\mathrm{F}}  \leq c_2 \frac{r \epsilon}{m} ,$$
with probability exceeding $1-\exp(-\gamma m/r^{2})$ for some constants $c_2$ and $\gamma$.
\end{theorem}

Theorem~\ref{main} has the following consequences.
\begin{itemize}
\item \textbf{Exact Recovery with Outliers:} When $\epsilon=0$, Theorem~\ref{main} suggests the recovery is exact using Robust-PhaseLift \eqref{phaselift_outlier}, i.e. $\hat{\bX}=\bX_0$ even when a fraction of measurements are arbitrarily corrupted, as long as the number of measurements $m$ is on the order of $nr^{2}$. Given there are at least $nr$ unknowns, our measurement complexity is near-optimal up to a factor of $r$.
\item \textbf{Stable Recovery with Bounded Noise:} In the presence of bounded noise, Theorem~\ref{main} suggests that the recovery performance decreases gracefully with the increase of $\epsilon$, where the Frobenius norm of the reconstruction error is proportional to the per-entry noise level of the measurements.
\item \textbf{Phase Retrieval:} When $r=1$, the problem degenerates to the case of phase retrieval, and Theorem~\ref{main} recovers existing results in \cite{hand2016phaselift} for outlier-robust phase retrieval, where the measurement complexity is on the order of $n$, which is optimal up to a scaling factor.
\end{itemize}

\yc{Let us denote $\hat{\boldsymbol{X}}_r=\argmin_{\mbox{rank}(\boldsymbol{Z})=r,\boldsymbol{Z}\succeq 0} \| \hat{\boldsymbol{X}}- \boldsymbol{Z}\|_{\mathrm{F}}$ as the best rank-$r$ PSD matrix approximation of $\hat{\bX}$, the solution to \eqref{phaselift_outlier}. Then Theorem~\ref{main} suggests that the estimate $\hat{\bX}$ can be well approximated by a rank-$r$ PSD matrix since
 $$ \| \hat{\boldsymbol{X}} - \hat{\boldsymbol{X}}_r \|_F \leq \| \hat{\boldsymbol{X}} - \boldsymbol{X}_0 \|_F\leq c_2 \frac{r\epsilon}{m},$$
as long as the number of measurements is sufficiently large. Furthermore, we have
\begin{align*}
\|\hat{\boldsymbol{X}}_r - \boldsymbol{X}_0 \|_F  & \leq \|\hat{\boldsymbol{X}}_r - \hat{\boldsymbol{X}}\|_F + \|\hat{\boldsymbol{X}} - \boldsymbol{X}_0 \|_F \\
& \leq 2 \|\hat{\boldsymbol{X}} - \boldsymbol{X}_0 \|_F \leq 2c_2\frac{r\epsilon}{m},
\end{align*}
indicating that $\hat{\boldsymbol{X}}_r $ provides an accurate estimate of $\bX_0$ that is exactly rank-$r$ and PSD.}

\subsection{Comparisons to Related Work}

In the absence of outliers, the PhaseLift algorithm in the following form
\begin{equation}\label{phaselift_constraint}
 \min_{\bX\succeq 0} \mbox{Tr}(\boldsymbol{X}) \quad \mbox{s.t.} \quad \|\boldsymbol{z} - \mathcal{A}(\bX)\|_1\leq \epsilon ,
\end{equation}
where $\mbox{Tr}(\bX)$ denotes the trace of $\bX$, has been proposed to solve the phase retrieval problem \cite{candes2013phaselift,candes2012solving,candes2013phase}. Later the same algorithm has been employed to recover low-rank PSD matrices in \cite{chen2015exact}, where an order of $nr$ measurements obtained from i.i.d. sub-Gaussian sensing vectors are shown to guarantee exact recovery in the noise-free case and stable recovery with bounded noise. One problem with the algorithm \eqref{phaselift_constraint} is that the noise bound $\epsilon$ is assumed known. Furthermore, it is not amenable to handle outliers, since $\| \boldsymbol{z} - \cA(\bX_0)\|_1$ can be arbitrarily large with outliers and consequently the ground truth $\bX_0$ quickly becomes infeasible for \eqref{phaselift_constraint}. 

The proposed algorithm \eqref{phaselift_outlier} is studied in \cite{candes2012solving,demanet2012stable,hand2016phaselift} as a variant of PhaseLift for phase retrieval, corresponding to the case where $\bX_0=\bx_0\bx_0^T$ is rank-one. It is shown in \cite{demanet2012stable,candes2012solving} that with $\mathcal{O}(n)$ i.i.d. Gaussian sensing vectors, the algorithm succeeds with high probability. Compared with \eqref{phaselift_constraint}, the algorithm \eqref{phaselift_outlier} eliminates trace minimization and leads to easier algorithm implementations. We note that \cite{kabanava2015stable} also considers a regularization-free algorithm for PSD matrix estimation that minimizes the $\ell_2$-norm of the residual, which unfortunately, cannot handle outliers as Robust-PhaseLift \eqref{phaselift_outlier}. Hand \cite{hand2016phaselift} first considered the robustness of the Robust-PhaseLift algorithm \eqref{phaselift_outlier} in the presence of outliers for phase retrieval, establishing that the same guarantee holds even with a constant fraction of outliers. Our work extends the performance guarantee in \cite{hand2016phaselift} to the general low-rank PSD matrix case. 

Broadly speaking, our problem is related to low-rank matrix recovery from an under-determined linear system \cite{recht2010guaranteed,dai2011subspace,wang2011unique}, where the linear measurements are drawn from inner products with rank-one sensing matrices. It is due to this special structure of the sensing matrices that we can eliminate the trace minimization, and only consider the feasibility constraint for PSD matrices. Standard approaches for separating low-rank and sparse components \cite{CanLiMaWri09,chandrasekaran2011rank,wright2013compressive,li2011compressed,mateos2012robust} via convex optimization are given as
\begin{equation*}
\min_{\bX\succeq 0,\ \bbeta} \mbox{Tr}(\bX) + \lambda \|\bbeta\|_1, \quad \mbox{s.t.} \quad \|\bz- \mathcal{A}(\bX)-\bbeta\|_1 \leq \epsilon,
\end{equation*}
where $\lambda$ is a regularization parameter that requires to be tuned properly. In contrast, the formulation \eqref{phaselift_outlier} is parameter-free.


\begin{figure*}[ht]
\begin{center}
\begin{tabular}{ccc}
 {no outliers} &modest outlier amplitudes &large outlier amplitudes \\
\includegraphics[width=0.3\textwidth]{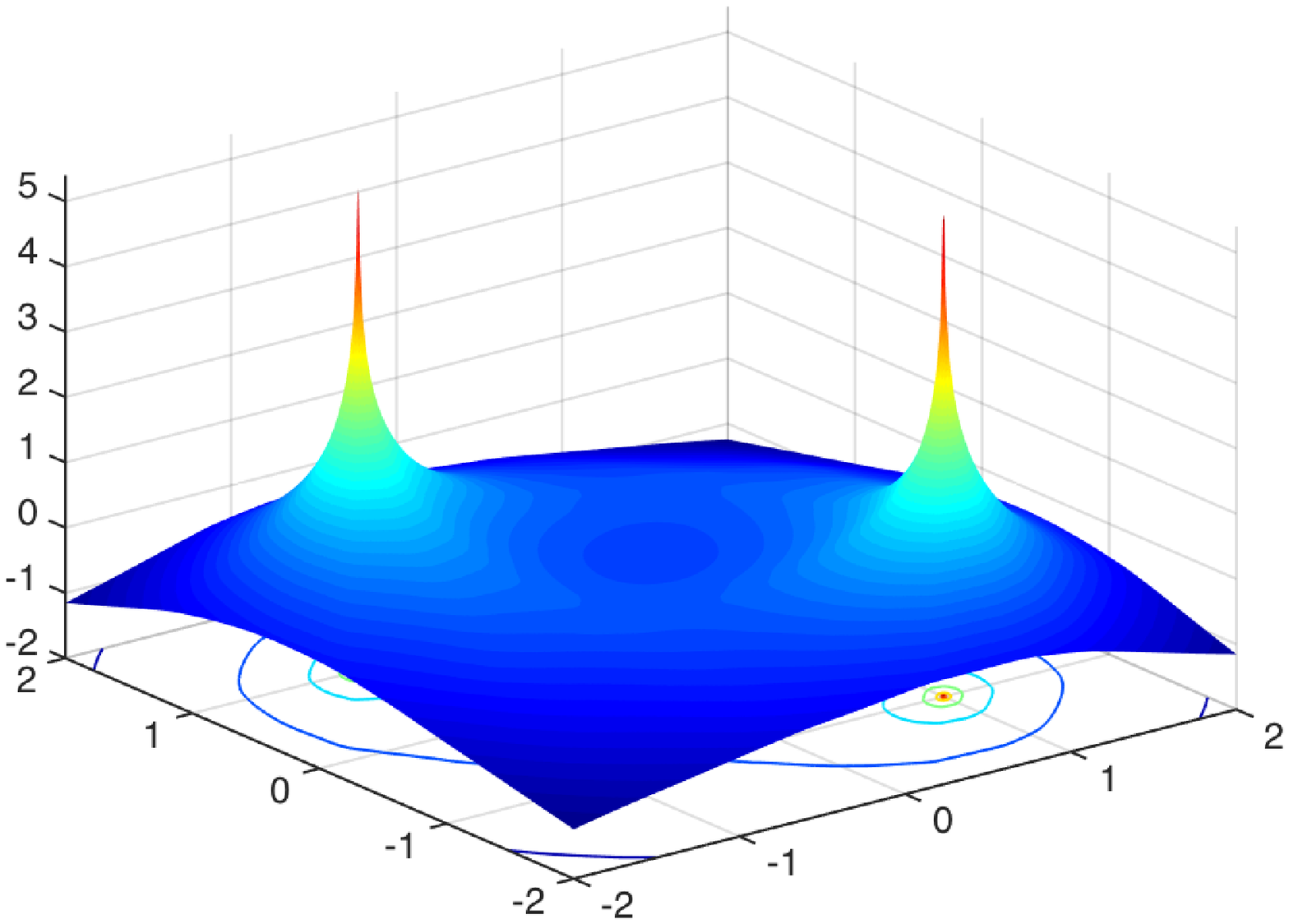} &\includegraphics[width=0.3\textwidth]{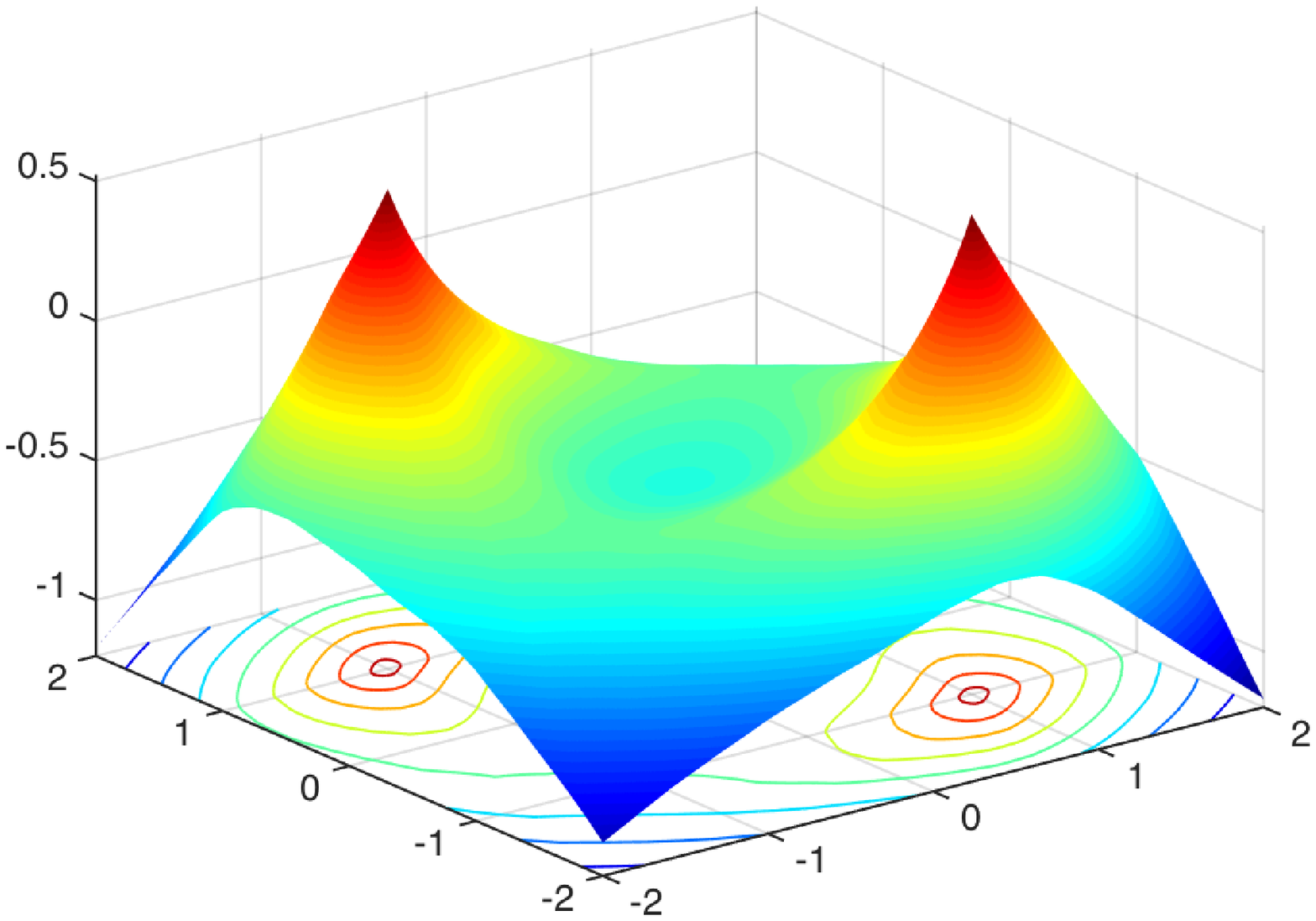} & \includegraphics[width=0.3\textwidth]{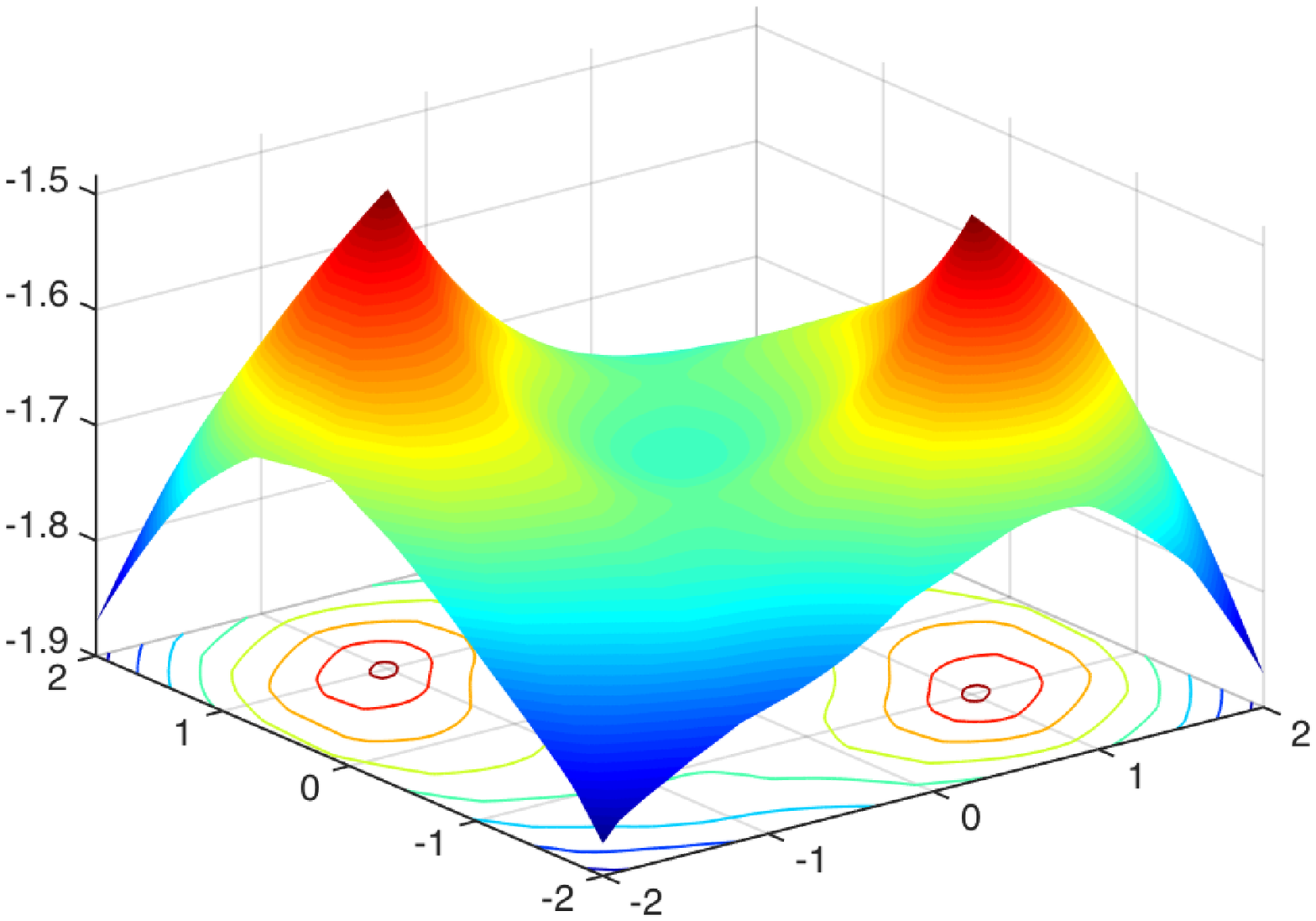}  \\
& $f(\bU)= \frac{1}{2m} \| \bz - \mathcal{A}(\bU\bU^T)\|_1$ & \\
\includegraphics[width=0.3\textwidth]{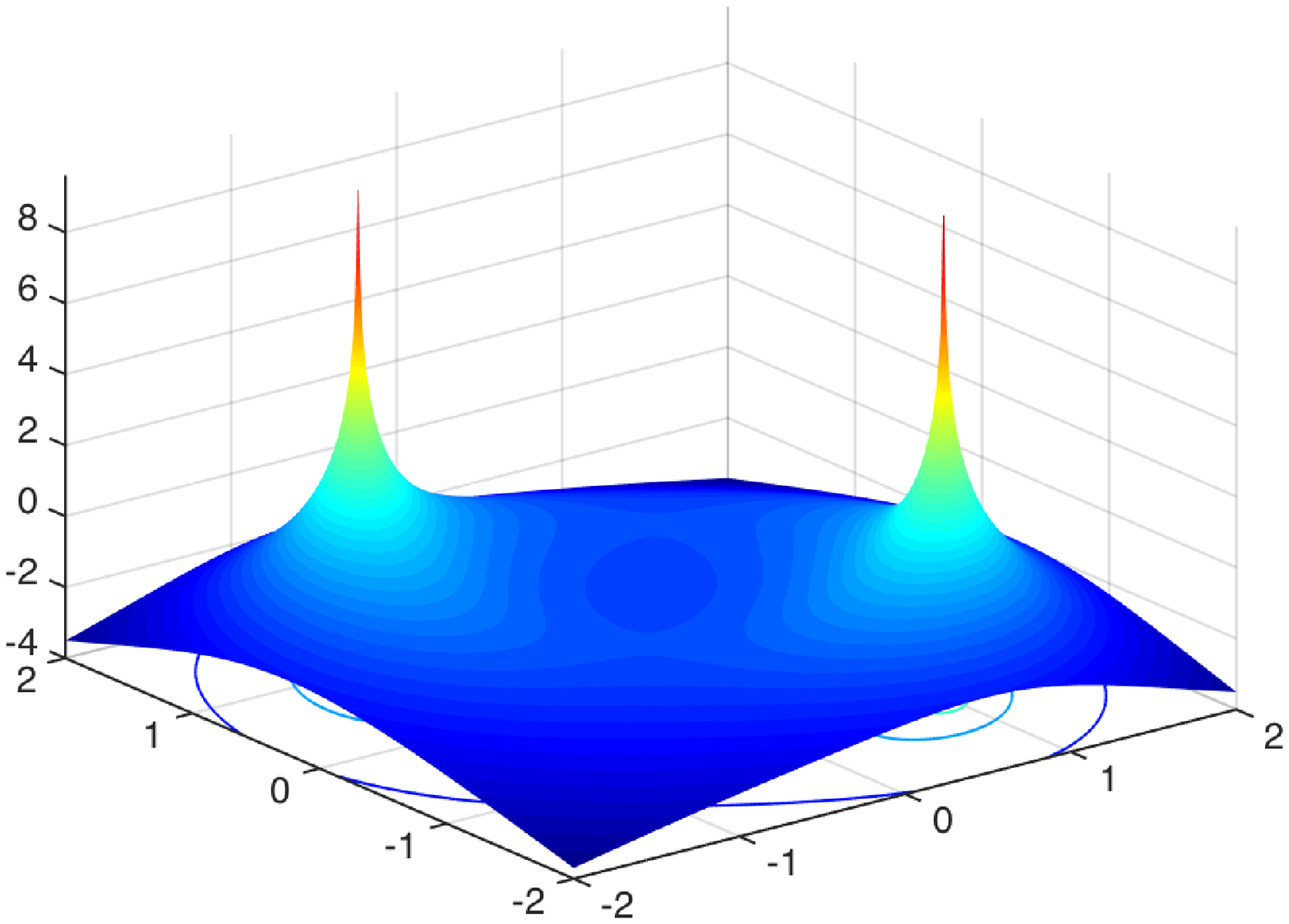}  &
\includegraphics[width=0.3\textwidth]{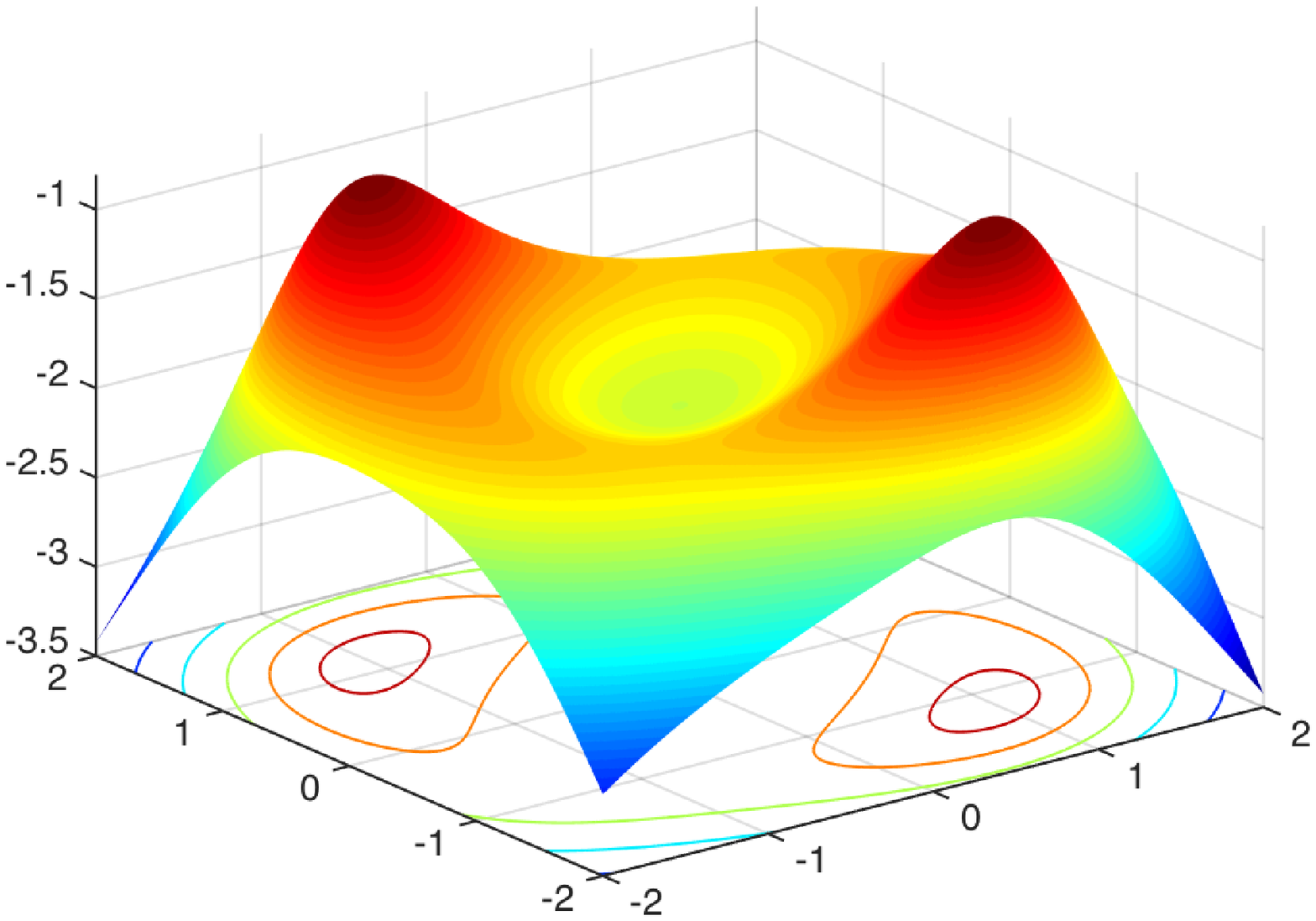} &
\includegraphics[width=0.3\textwidth]{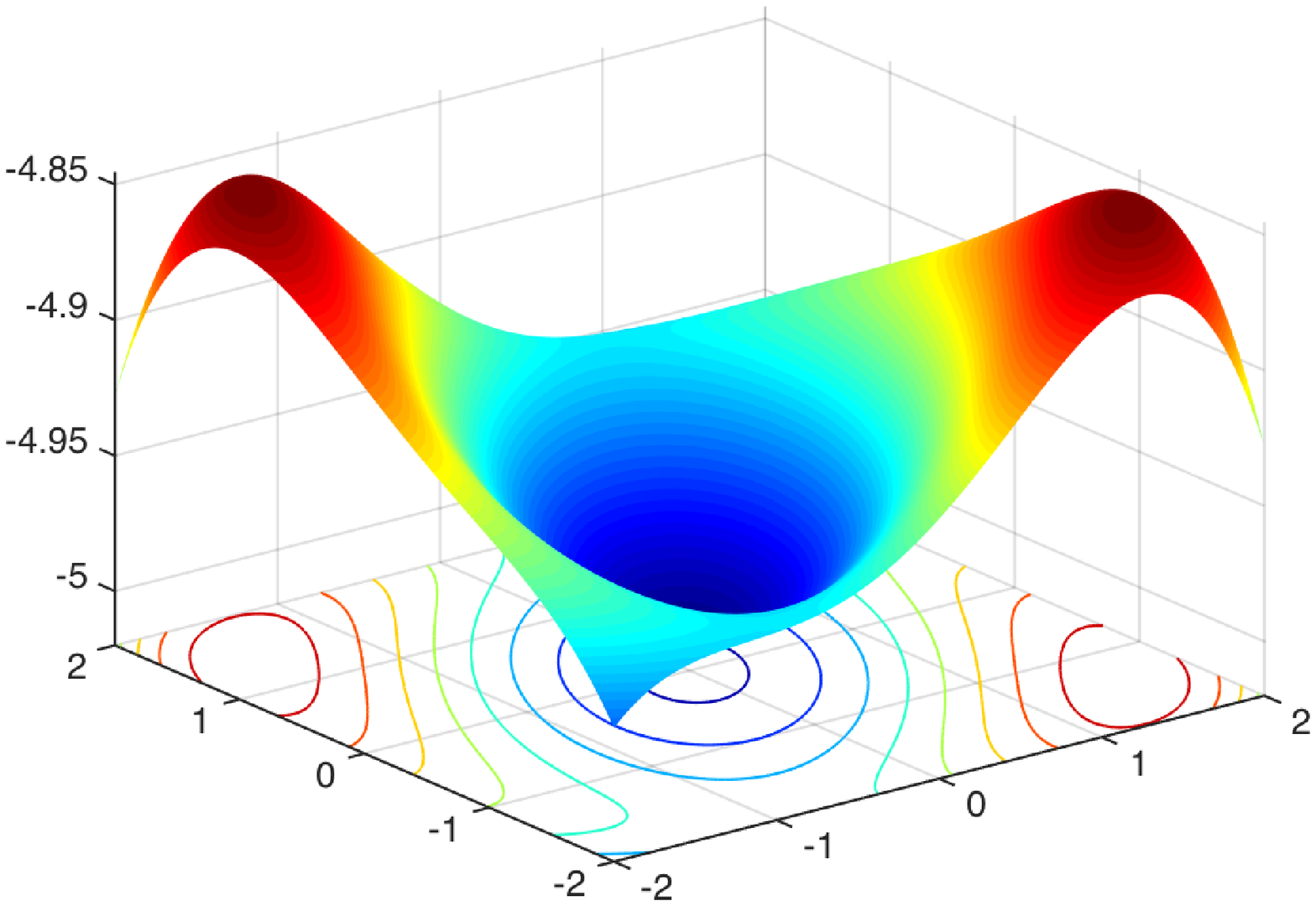} \\
& $g(\bU)= \frac{1}{4m} \| \bz - \mathcal{A}(\bU\bU^T)\|_2^2$ &
\end{tabular}
\end{center}
\caption{Illustrations of the objective function $-\log f(\bU)$ and its $\ell_2$-norm counterpart $-\log g(\bU)$ (in negative logarithmic scales) under different corruption scenarios when $\bU\in\mathbb{R}^{2\times 1}$. The number of measurements is $m=100$ with i.i.d. Gaussian sensing vectors, and the fraction of outliers is $s=0.2$ with uniformly selected support and amplitudes drawn from $\mbox{Unif}[0,10]$ or $\mbox{Unif}[0,100]$. It is interesting to observe that while large outliers completely distort $g(\bU)$, the proposed objective is quite robust with the ground truth being the only global optima of $f(\bU)$. }\label{fig:illustration_nonconvex}
\end{figure*}

\section{A Non-Convex Subgradient Descent Algorithm}\label{sec:non_convex}

\yc{In this section, we propose another algorithm for robust low-rank PSD matrix recovery from corrupted rank-one measurements assuming the rank (or its upper bound) of the PSD matrix $\bX_0$ is known a priori as $r$.} In this case, we can decompose $\bX_0$ as $\bX_0 = \bU_0\bU_0^T$ where $\bU_0\in\mathbb{R}^{n\times r}$. Instead of directly recovering $\bX_0$, we may aim at recovering $\bU_0$ up to orthogonal transforms, since $(\bU_0\bQ)(\bU_0\bQ)^T =\bU_0\bU_0$ for any orthonormal matrix $\bQ\in\mathbb{R}^{r\times r}$. Consider relaxing of the loss function in \eqref{non_convex_outliers} but keeping the rank constraint, we obtain the following problem:
\begin{equation}\label{eq:non_convex_psd} 
\hat{\bX} = \argmin_{\bX\succeq 0}\|\boldsymbol{z} - \mathcal{A}(\bX)\|_1, \quad \mbox{s.t.} \quad \mbox{rank}(\bX) = r.
\end{equation}
Since any rank-$r$ PSD matrix $\bX$ can be written as $\bX = \bU\bU^T$ for some $\bU\in\mathbb{R}^{n\times r}$, 
 \eqref{eq:non_convex_psd} can be equivalently reformulated as
\begin{equation}\label{phaselift_outlier_nonconvex}
\hat{\bU} = \argmin_{\bU\in\mathbb{R}^{n\times r}}  f(\bU)   ,
\end{equation} 
with
\begin{equation*}
f(\bU) = \frac{1}{2m} \left\Vert \boldsymbol{z} - \mathcal{A}(\bU\bU^T)\right\Vert_{1} =\frac{1}{2m}\sum_{i=1}^m  \left\vert z_i - \left\Vert \bU^T\ba_i\right\Vert_2^2 \right\vert.
\end{equation*}
Clearly, \eqref{phaselift_outlier_nonconvex} is no longer convex. To illustrate, the first row of Fig.~\ref{fig:illustration_nonconvex} plots the value of the objective function in the negative logarithmic scale, i.e. $-\log f(\bU)$, under different corruption scenarios when $\bU\in\mathbb{R}^{2\times 1}$. For comparison, the second row of Fig.~\ref{fig:illustration_nonconvex} shows the loss function evaluated in $\ell_2$-norm: $g(\bU)= \frac{1}{4m} \| \bz - \mathcal{A}(\bU\bU^T)\|_2^2$, which is not robust to outliers.

Motivated by the recent non-convex approaches \cite{candes2015phase,ChenCan15,white2015local} of solving quadratic systems, we propose a subgradient descent algorithm to solve \eqref{phaselift_outlier_nonconvex} effectively, working with a non-smooth function $f(\bU)$. Note that a subgradient of $f(\bU)$ with respect to $\bU$ can be given as
\begin{equation}
\partial f(\bU)=-\frac{1}{m}\sum_{i = 1}^m \sgn\left( z_i - \left\Vert \bU^T\ba_i\right\Vert_2^2\right) \ba_i\ba_i^T\bU,
\end{equation}
where the sign function $\sgn(\cdot)$ is defined as
$$ \sgn(x) = \left\{ \begin{array}{cc}
+1, & x>0 \\
0, & x=0 \\
-1, & x<0 \end{array}\right. .$$

Our subgradient descent algorithm proceeds as below. Denote the estimate in the $t$th iteration by $\bU^{(t)}\in\mathbb{R}^{n\times r}$. First, $\bU^{(0)}$ is initialized as the best rank-$r$ approximation of the following matrix with respect to Frobenius norm as
\begin{equation}\label{initialization}  
\bU^{\left(0\right)} \left(\bU^{\left(0\right)}\right)^{T} = \argmin_{\mathrm{rank}\left(\bX\right) = r} \left\Vert \bX - \frac{1}{m}\sum_{i=1}^{m} z_{i} \ba_{i}\ba_{i}^{T} \right\Vert_{\mathrm{F}}^{2}.
\end{equation}
Secondly, at the $(t+1)$th iteration, $t\geq 0$, we apply subgradient descent to refine the estimate as
\begin{equation}\label{iteration}
\bU^{(t+1)} = \bU^{(t)} - \mu_t \cdot \partial f(\bU^{(t)}) ,
\end{equation}
where the step size $\mu_{t}$ is adaptively set as
\begin{equation*}
\mu_{t} = 0.05\times  \max \left\{2^{-t/1000},  10^{-6}\right\},
\end{equation*}
which provide more accurate estimates using fewer iterations in the numerical simulations. The procedure is summarized in Alg.~\ref{algorithm:nonconvex}, 
where the stopping rule in Alg.~\ref{algorithm:nonconvex} is simply put as a maximum number of iterations. 


\begin{algorithm}[htp]
  \caption{Subgradient descent for solving \eqref{phaselift_outlier_nonconvex}}\label{algorithm:nonconvex}
  \textbf{Parameters:} Rank $r$, number of iterations $T_{\max}$, and step size $\mu_t$;\\
  
  \textbf{Input:} Measurements $\bz$, and sensing vectors $\{\ba_i\}_{i=1}^m$;\\

\textbf{Initialization:} Initialize $\bU^{(0)}\in\mathbb{R}^{n\times r}$ via \eqref{initialization}; \\

  \textbf{for} $t=0:T_{\max}-1$
   \textbf{do} \\
   \hspace{0.2in} update $\bU^{(t+1)}$ via \eqref{iteration}; \\
 \textbf{end for}\\
\textbf{Output:} $\hat{\bU}= \bU^{(T_{\max})}$.
\end{algorithm}

The main advantage of Alg.~\ref{algorithm:nonconvex} is its low memory and computational complexity. Given that it does not construct the full PSD matrix, the memory complexity is simply the size of $\bU^{(t)}$, which is on the order of $nr$. The computational complexity per iteration is also low, which is on the order of $mnr$, that is linear in all the parameters. We demonstrate the excellent empirical performance of Alg.~\ref{algorithm:nonconvex} in Section~\ref{sec:performance_nonconvex}.

\section{Numerical Examples} \label{sec:numerical}

\subsection{Performance of Convex Relaxation}
We first examine the performance of Robust-PhaseLift in \eqref{phaselift_outlier}. Let $n=40$. We randomly generate a low-rank PSD matrix of rank-$r$ as $\bX_{0}=\bU_{0}\bU_{0}^{T}$, where $\bU_{0}\in\mathbb{R}^{n\times r}$ is composed of i.i.d. standard Gaussian variables. The sensing vectors are also composed of i.i.d. standard Gaussian variables. Each Monte Carlo simulation is called successful if the normalized estimate error satisfies $\|\hat{\bX}-\bX_{0}\|_{\mathrm{F}}/\|\bX_{0}\|_{\mathrm{F}}\leq 10^{-6}$, where $\hat{\bX}$ denotes the solution to \eqref{phaselift_outlier}. For each cell, the success rate is calculated by averaging over $100$ Monte Carlo simulations.
\begin{figure}[h]
\begin{center}
\begin{tabular}{cc}
\hspace{-0.1in}\includegraphics[width=0.25\textwidth]{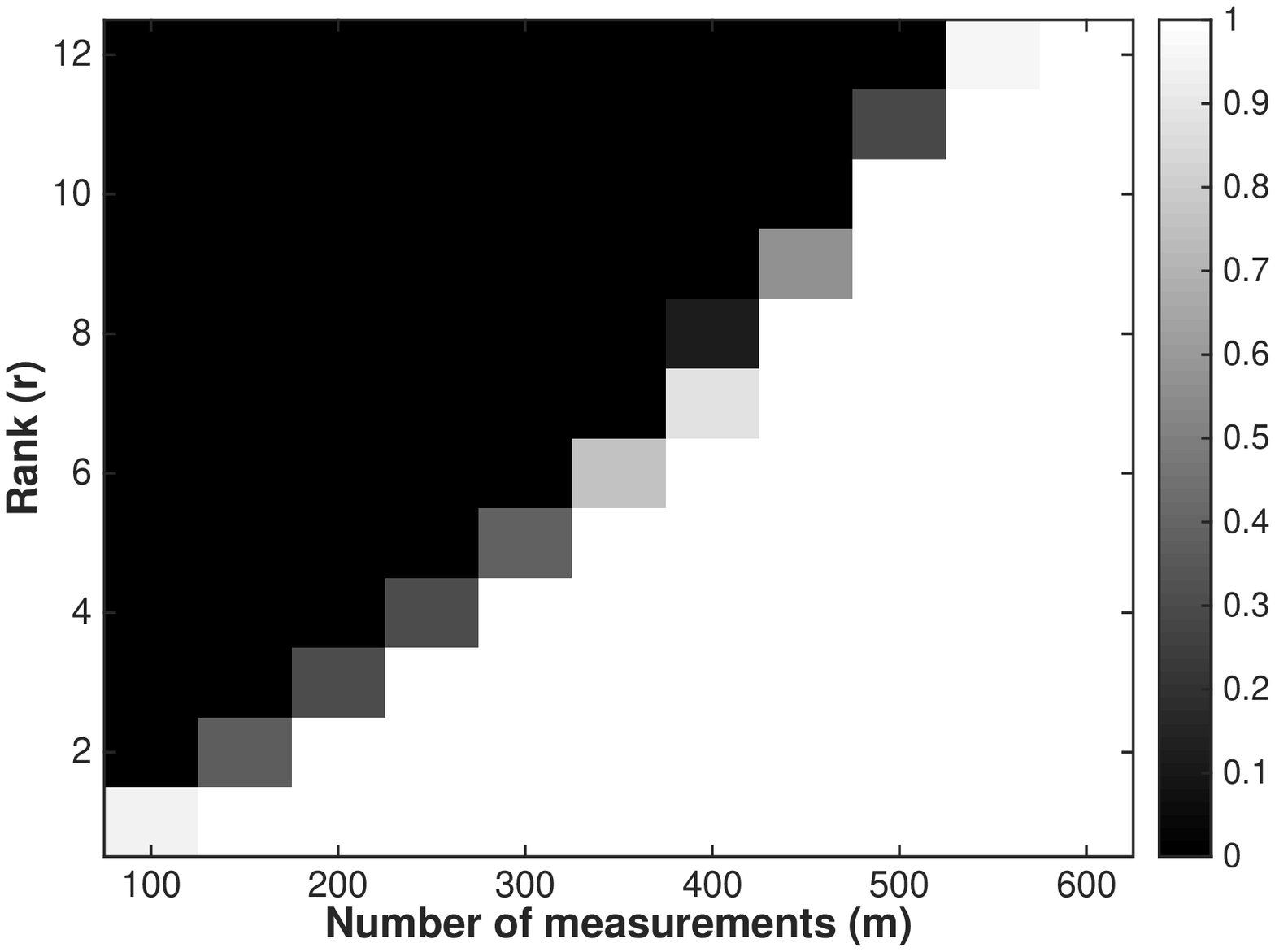} &
\hspace{-0.2in}\includegraphics[width=0.25\textwidth]{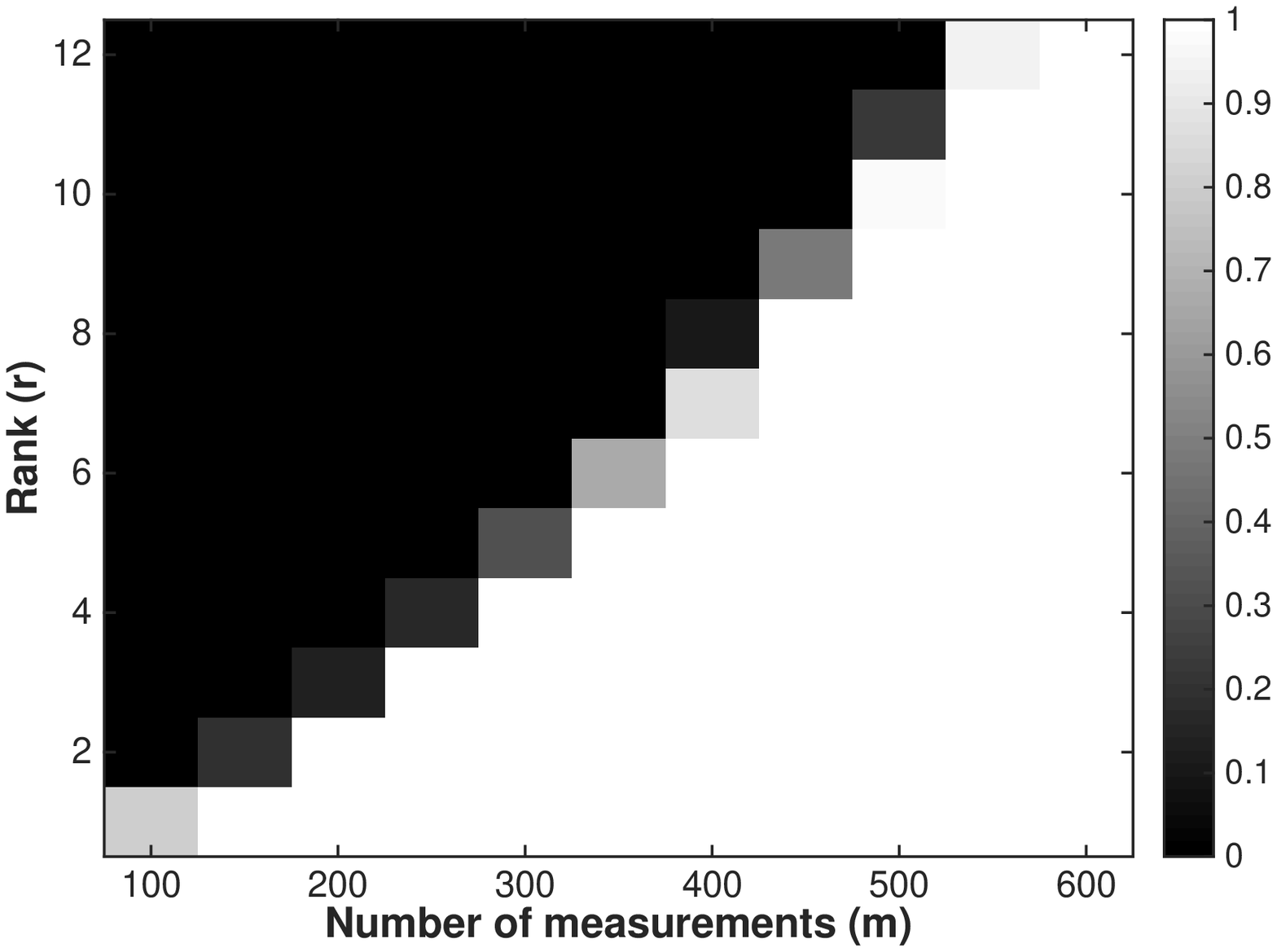} \\
\hspace{-0.1in}(a)   & \hspace{-0.2in} (b)  
\end{tabular}
\end{center}
\caption{Phase transitions for low-rank PSD matrix recovery with respect to the number of measurements and the rank, (a) with trace minimization; and (b) without trace minimization of noise-free measurements, when $n=40$.   }\label{fig_lowrank_psd_rec_fixn_changemr_psdtrace_noutnbd}
\end{figure}

Fig.~\ref{fig_lowrank_psd_rec_fixn_changemr_psdtrace_noutnbd} shows the success rates of algorithms with respect to the number of measurements and the rank, with the trace minimization as in \eqref{phaselift_constraint} in (a); and without the trace minimization as proposed in Robust-PhaseLift \eqref{phaselift_outlier} in (b) for noise-free measurements. It can be seen that the performance of these two algorithms are almost equivalent, confirming a similar numerical observation for the phase retrieval problem \cite{waldspurger2015phase} also holds in the low-rank setting, where trace minimization may be eliminated for low-rank PSD matrix recovery using rank-one measurements.

Fig.~\ref{fig_lowrank_psd_rec_psd} further shows the success rates of the Robust-PhaseLift algorithm (a) with respect to the number of measurements and the rank, when 5\% of measurements are selected uniformly at random and corrupted by standard Gaussian variables; and (b) with respect to the percent of outliers and the rank, for a fixed number of measurements $m=600$. This also suggests possible room for improvements of our theoretical guarantee, as the numerical results indicate that the required measurement complexity for successful recovery has a seemingly linear relationship with $r$.

\begin{figure}[h]
\begin{center}
\begin{tabular}{cc}
\hspace{-0.1in}\includegraphics[width=0.25\textwidth]{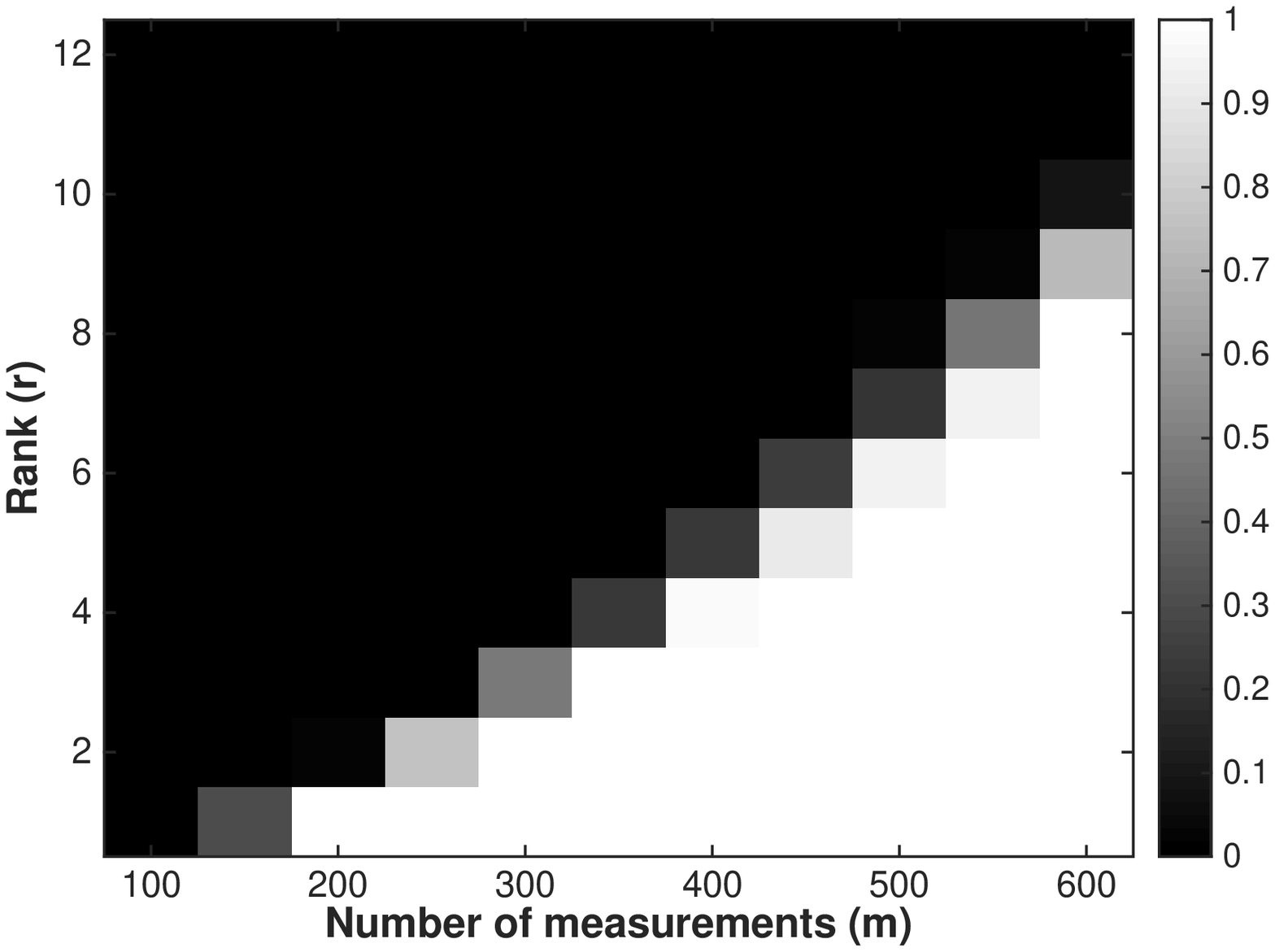} &
\hspace{-0.2in}\includegraphics[width=0.25\textwidth]{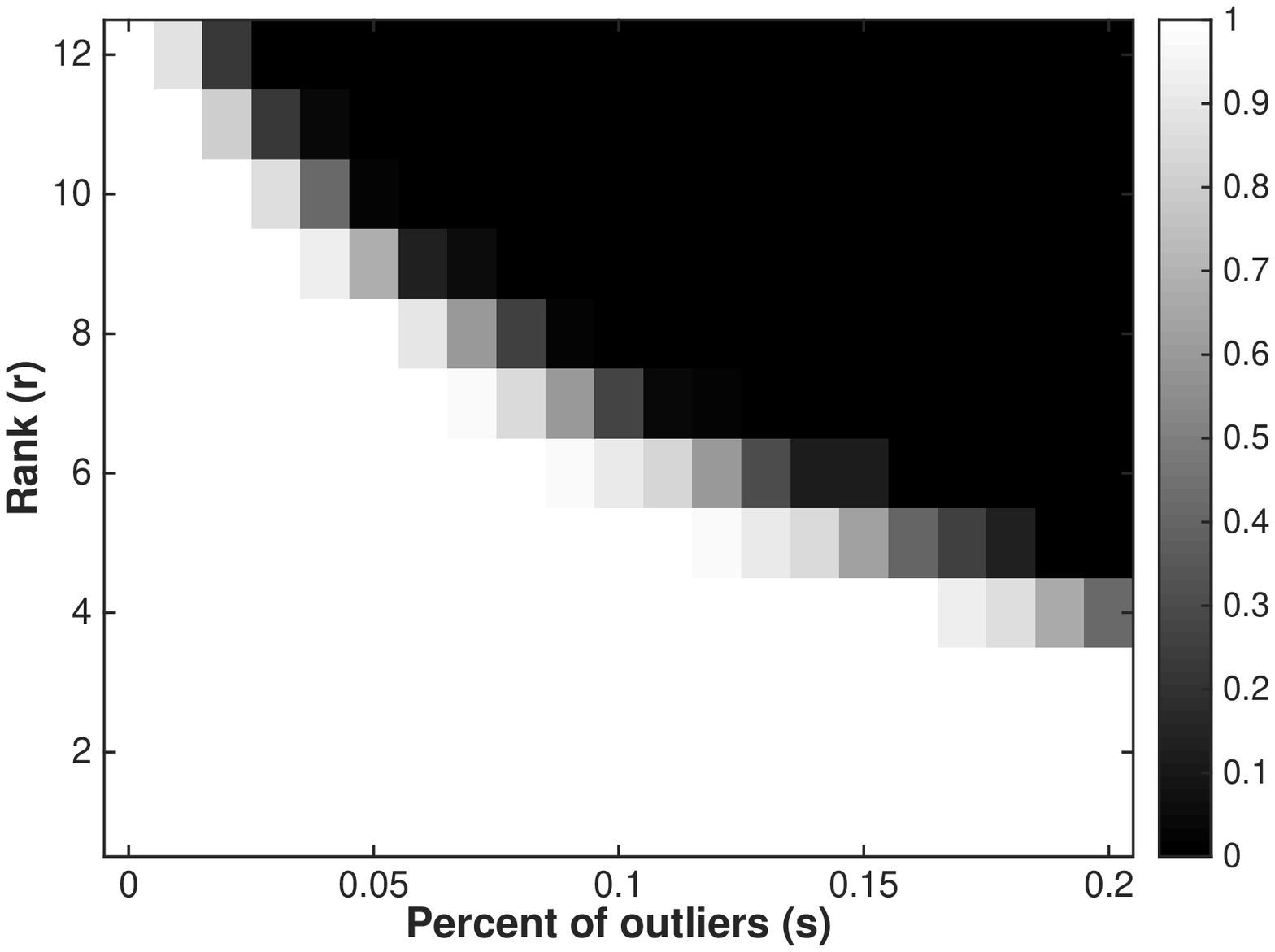} \\
\hspace{-0.1in} (a) & \hspace{-0.2in} (b)
\end{tabular}
\end{center}
\caption{Phase transitions of low-rank PSD matrix recovery with respect to (a) the number of measurements and the rank, with 5\% of measurements corrupted by standard Gaussian variables; (b) the percent of outliers and the rank, when the number of measurements is $m=600$, when $n=40$.}\label{fig_lowrank_psd_rec_psd}
\end{figure}

\subsection{Convex Relaxation with additional Toeplitz Structure}
We next consider robust recovery of low-rank Toeplitz PSD matrices, where we allow complex-valued sensing vectors $\mathcal{A}(\bX) = \{ \ba_i^{H}\bX\ba_i \}_{i=1}^m$ and complex-valued Toeplitz PSD matrices $\bX$. Estimating low-rank Toeplitz PSD matrices is of great interests for array signal processing \cite{abramovich1998positive}. We modify \eqref{phaselift_outlier} by incorporating the Toeplitz constraint as:
\begin{equation}\label{toeplitz_phaselift_outlier}
\hat{\bX} = \argmin_{\bX\succeq 0} \|\boldsymbol{z} - \mathcal{A}(\bX)\|_1,  \; \mbox{s.t.} \;    \bX~\mbox{is Toeplitz}.
\end{equation} 
Let $n=64$, the Toeplitz PSD matrix $\bX_0$ is generated as $\bX_0 =\bV\bSigma\bV^H$, where $\bV=[\bv(f_1),\ldots,\bv(f_r)]\in\mathbb{C}^{n\times r}$ is a Vandermonde matrix with $\bv(f_i)=[1,e^{j2\pi f_i},\ldots, e^{j2\pi (n-1)f_i}]^T$, $f_i\sim \mbox{Unif}[0,1]$, and $\bSigma=\mbox{diag}[\sigma_1^2,\ldots, \sigma_r^2]$, with $\sigma_i^2\sim  \mbox{Unif}[0,1]$. 
Fig.~\ref{fig:rank_vs_outliers_toeplitz} shows the phase transitions of Toeplitz PSD matrix recovery with respect to the number of measurements and the rank without outliers in (a), and when 5\% of measurements are selected uniformly at random and corrupted by standard Gaussian variables in (b). It can be seen that the low-rank Toeplitz PSD matrix can be robustly recovered from a sublinear number of measurements due to the additional Toeplitz structure. \yc{We note that a different covariance sketching scheme is considered in \cite{qiao2015generalized,romero2016compressive,romero2015compression} for estimating low-rank Toeplitz covariance matrices. Though not directly comparable to our measurement scheme, it may benefit from a similar parameter-free convex optimization to handle outliers.}
\begin{figure}[h]
\begin{center}
\begin{tabular}{cc}
\hspace{-0.1in}\includegraphics[width=0.25\textwidth]{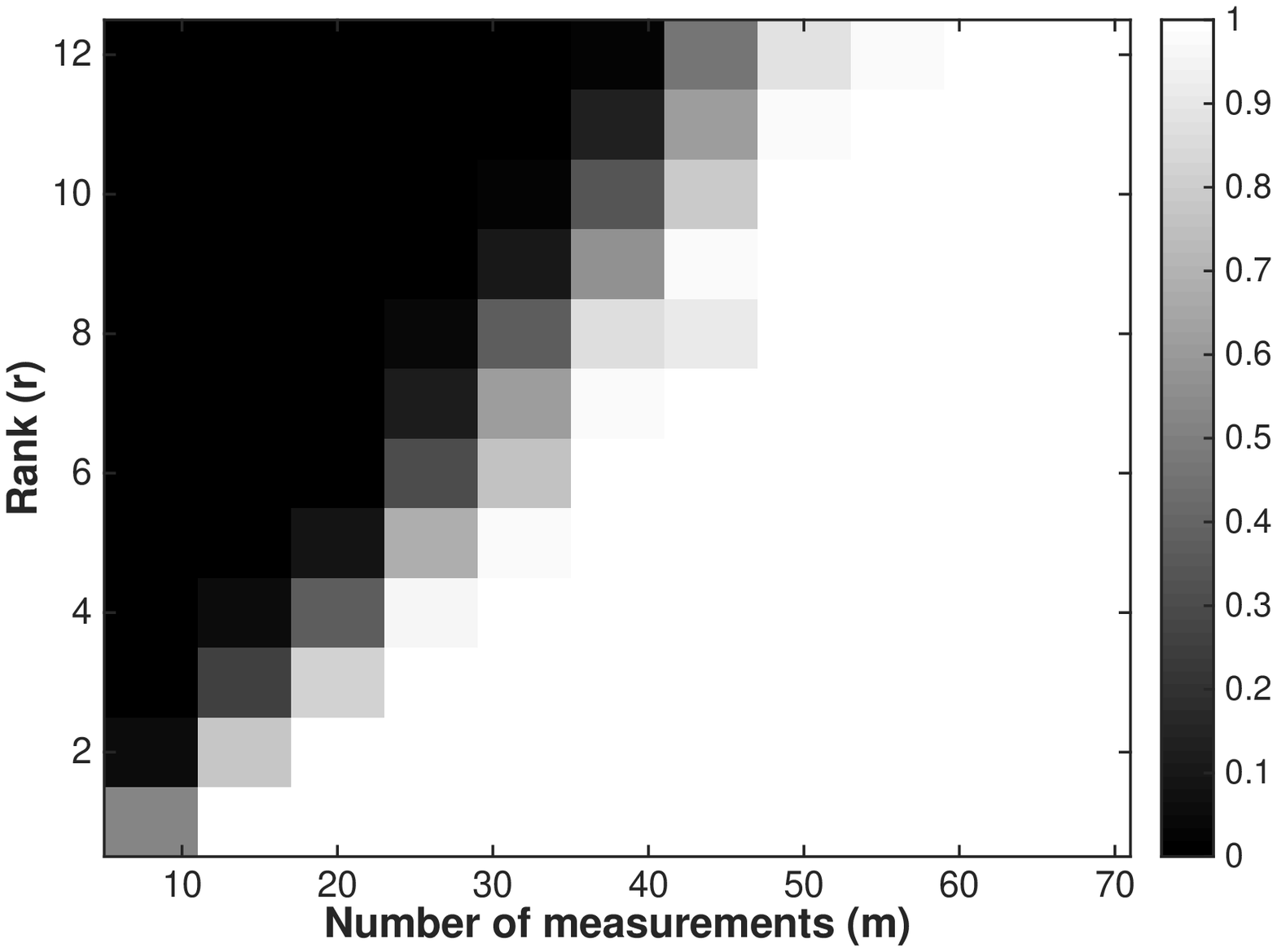} &
\hspace{-0.2in}\includegraphics[width=0.25\textwidth]{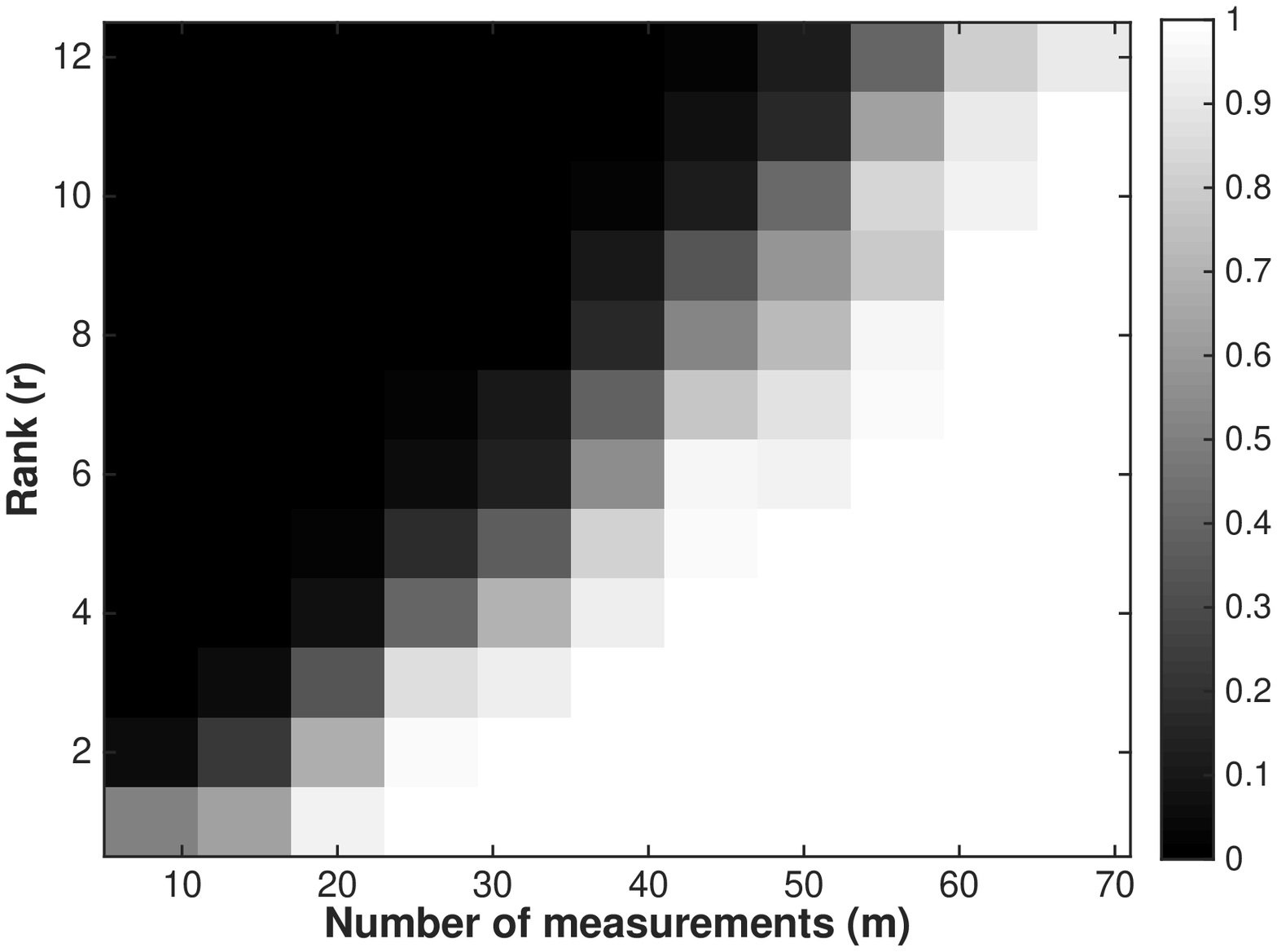} \\
\hspace{-0.1in}(a) & \hspace{-0.2in}(b)
\end{tabular}
\end{center}
\caption{Phase transitions of low-rank Toeplitz PSD matrix recovery with respect to the number of measurements and the rank, (a) without outliers, and (b) with 5\% of measurements corrupted by standard Gaussian variables, when $n=64$.   }\label{fig:rank_vs_outliers_toeplitz}
\end{figure}

\subsection{Performance of Non-Convex Subgradient Descent} \label{sec:performance_nonconvex}

We next examine the performance of the non-convex subgradient descent algorithm in Alg.~\ref{algorithm:nonconvex}, where the number of iterations is set as $T_{\max} = 3\times 10^4$, which is a large value to guarantee convergence when terminated. Denote the solution to Alg.~\ref{algorithm:nonconvex} by $\hat{\bU}$, and each Monte Carlo simulation is deemed successful if the normalized estimate error satisfies $\|\hat{\bX}-\bX_{0}\|_{\mathrm{F}}/\|\bX_{0}\|_{\mathrm{F}}\leq 10^{-6}$, where $\hat{\bX}=\hat{\bU}\hat{\bU}^T$ is the estimated low-rank PSD matrix. For each cell, the success rate is calculated by averaging over $100$ Monte Carlo simulations.



\begin{figure}[h]
\begin{center}
\begin{tabular}{c}
\includegraphics[width=0.35\textwidth]{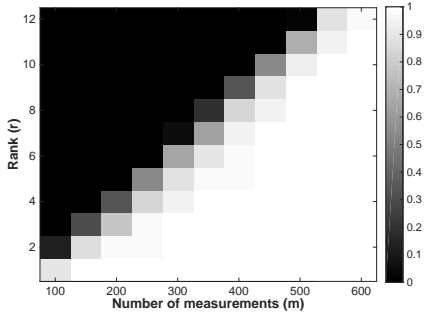} 
\end{tabular}
\end{center}
\caption{Phase transitions of low-rank PSD matrix recovery with respect to the number of measurements and the rank for the proposed Alg.~\ref{algorithm:nonconvex} using noise-free measurements, when $n=40$.}\label{fig_lowrank_psd_rec_fixn_changemr_L1L2_noutnbd}
\end{figure}

Fig.~\ref{fig_lowrank_psd_rec_fixn_changemr_L1L2_noutnbd} shows the success rate of Alg.~\ref{algorithm:nonconvex} with respect to the number of measurements and the rank under the same setup of Fig.~\ref{fig_lowrank_psd_rec_fixn_changemr_psdtrace_noutnbd} for noise-free measurements, when $n=40$. Indeed, empirically Alg.~\ref{algorithm:nonconvex} performs similarly as the convex algorithms but with a much lower computational cost. Moreover, the proposed Alg.~\ref{algorithm:nonconvex} allows perfect recovery even in the presence of outliers. For comparison, we implement the extension of the Wirtinger Flow (WF) algorithm in \cite{candes2015phase,white2015local,zheng2015convergent} in the low-rank case, that minimizes the squared $\ell_2$-norm of the residual, where the update rule per iteration becomes
\begin{equation*}
 \bU^{(t+1)}  =  \bU^{(t)} + \mu_t^{\mathrm{WF}} \frac{1}{m}\sum_{i = 1}^m \left( z_i -  \Vert  (\bU^{(t)} )^T\ba_i \Vert_2^2\right) \ba_i\ba_i^T\bU^{(t)}, 
 \end{equation*}
using the same initialization \eqref{initialization}. The step size is set as $\mu_t^{\mathrm{WF}} = 0.1/\left\Vert \bU_{0} \right\Vert_{\mathrm{F}}^{2}$. Fig. \ref{fig_lowrank_psd_rec_fixmn_changeroutlier_L1L2_nbd} (a) shows the success rates of Alg.~\ref{algorithm:nonconvex} with respect to the percent of outliers and the rank, under the same setup of Fig.~\ref{fig_lowrank_psd_rec_psd} (b), where the performance is even better than the convex counterpart in \eqref{phaselift_outlier}. In contrast, the WF algorithm performs poorly even with very few outliers, as shown in its success rate plot in Fig.~\ref{fig_lowrank_psd_rec_fixmn_changeroutlier_L1L2_nbd} (b), as the loss function used for WF is not robust to outliers.

\begin{figure}[h]
\begin{center}
\begin{tabular}{cc}
\hspace{-0.1in}\includegraphics[width=0.25\textwidth]{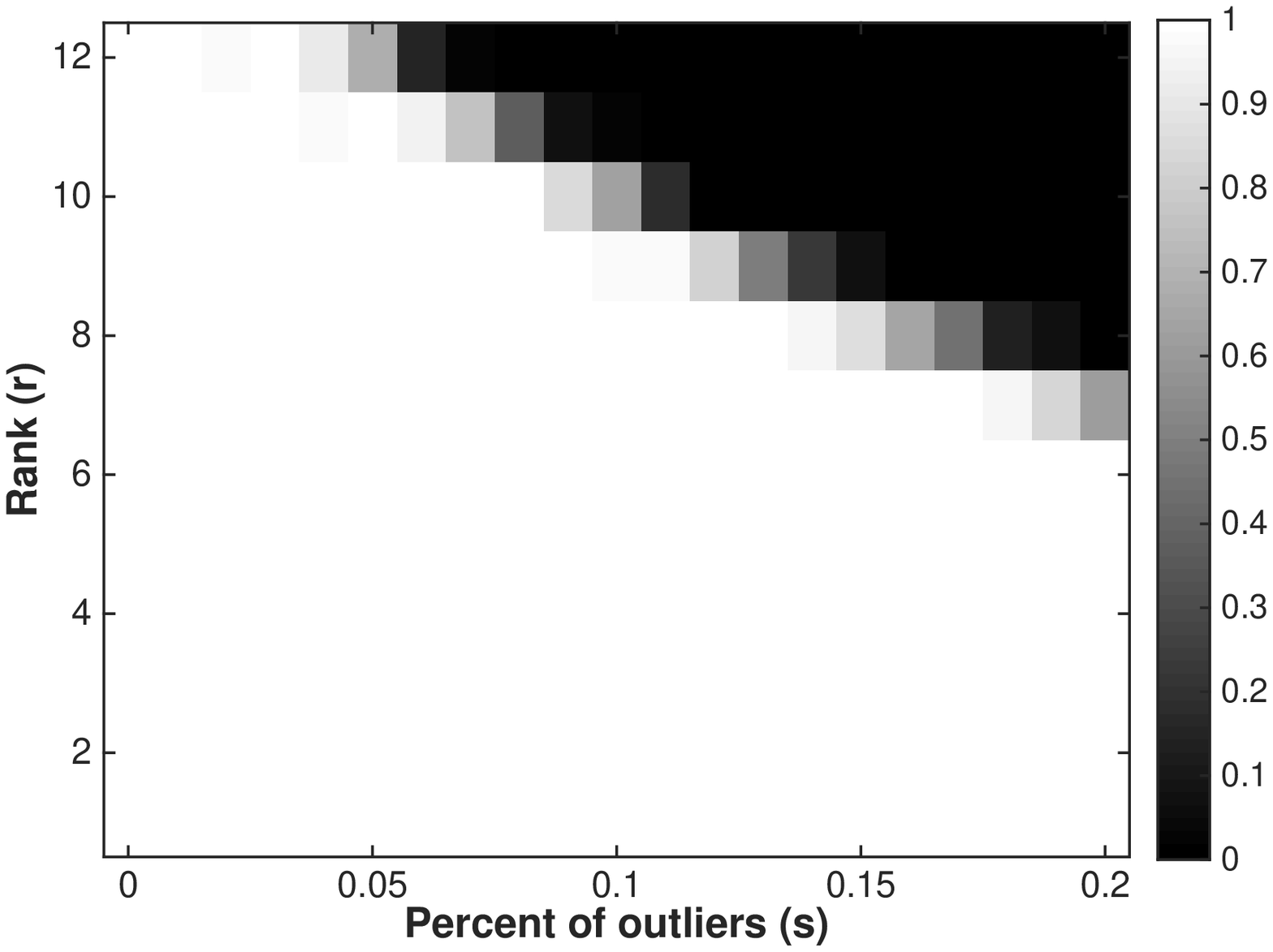} &
\hspace{-0.2in}\includegraphics[width=0.25\textwidth]{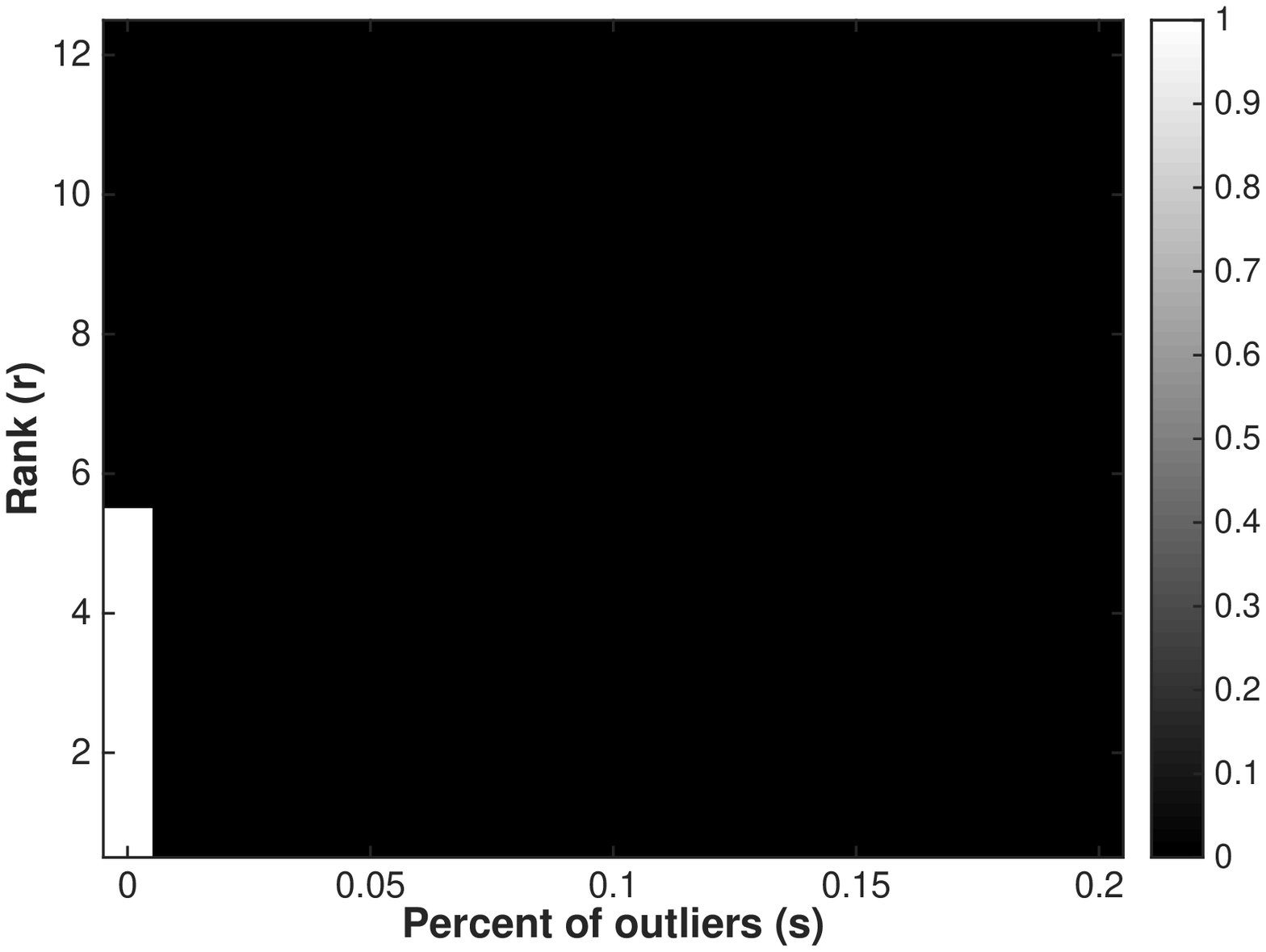} \\
\hspace{-0.1in}(a) & \hspace{-0.2in}(b)
\end{tabular}
\end{center}
\caption{Phase transitions of low-rank PSD matrix recovery with respect to the percent of outliers and the rank using (a) the proposed Alg.~\ref{algorithm:nonconvex}, and (b) the WF algorithm, when $n=40$ and $m=600$. }\label{fig_lowrank_psd_rec_fixmn_changeroutlier_L1L2_nbd}
\end{figure}


\subsection{Comparisons with Additional Bounded Noise}
 
\yxl{Finally, we compare the two proposed algorithms (Robust-PhaseLift in \eqref{phaselift_outlier} and Alg.~\ref{algorithm:nonconvex}), the WF algorithm and the PhaseLift algorithm in \eqref{phaselift_constraint} when the measurements are corrupted by both outliers and bounded noise. Fix $n=40$ and $r=3$. The rank-$r$ PSD matrix $\bX_{0}$, the sensing vectors, as well as the outliers are generated similarly as earlier, where the fraction of the outliers is set to $5\%$. Moreover, each entry in the bounded noise $\bw$ is i.i.d. drawn from $\mbox{Unif}[-4/m,4/m]$, thus $\| \bw \|_{1} \le \epsilon$, where $\epsilon = 4$. Fig.~\ref{fig_lowrank_psd_rec_fixnroutlier_changem_comp_mse} depicts the mean squared error $\| \hat{\bX} - \bX_{0} \|_{\mathrm{F}}^{2}$ for different algorithms with respect to the number of measurements, where $\hat{\bX}$ is the estimated PSD matrix. For the subgradient descent algorithm in Alg.~\ref{algorithm:nonconvex}, various ranks are used as prior information, corresponding to the correct rank $r$, its underestimate $r-1$, and its overestimate $r+1$. It can be seen that Alg.~\ref{algorithm:nonconvex} works well as long as the given rank provides an upper bound of the true rank, and it performs much better than the WF algorithm which is not outlier-robust. On the other hand, the PhaseLift algorithm \eqref{phaselift_constraint} does not admit favorable performance for various constraint parameters ($\epsilon$, $2\epsilon$, $4\epsilon$) as expected since the outliers do not fall into the prescribed noise bound. In fact, it fails to return any feasible solution when the number and amplitudes of outliers is too large in our simulation. In contrast, Robust-PhaseLift allows stable recovery even with an additional bounded noise, which performs comparably with Alg.~\ref{algorithm:nonconvex} with the correct model order.

\begin{figure}[h]
\begin{center}
\includegraphics[width=0.45\textwidth]{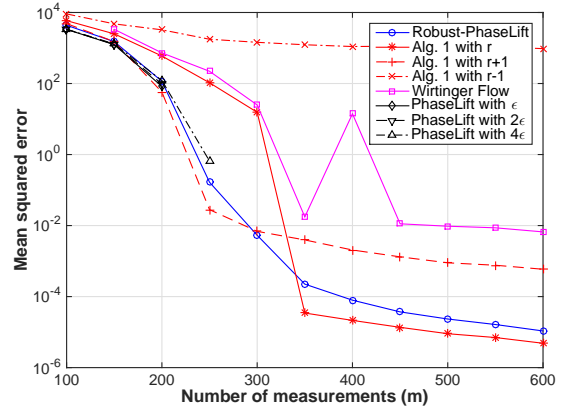}
\end{center}
\caption{Comparisons of mean squared errors using different algorithms with respect to the number of measurements with $5\%$ outliers and bounded noise, when $n=40$ and $r=3$.}\label{fig_lowrank_psd_rec_fixnroutlier_changem_comp_mse}
\end{figure}
}

\section{Proof of Main Theorem}\label{sec:proofs}

In this section we prove Theorem~\ref{main}, and the roadmap of our proof is below. In Section~\ref{appro_dual_cert}, we first provide the sufficient conditions for an approximate dual certificate that certifies the optimality of the proposed algorithm \eqref{phaselift_outlier} 
in Lemma~\ref{lemma-dual-certificate}. Section~\ref{sec:isometryA} records a few lemmas that show $\cA$ satisfies the required restricted isometry properties. Then, a dual certificate is constructed and validated for a fixed low-rank PSD matrix $\bX_0$ in Section~\ref{dual-construction}. Finally, the proof is concluded in Section~\ref{main_theorem_proof}.

First we introduce some additional notations. 
\yxl{ Let $\mathcal{S}$ be a subset of $\left\{1,2,\dots,m\right\}$, then $\mathcal{S}^\perp$ is the complement of $\mathcal{S}$ with respect to $\left\{1,2,\dots,m\right\}$. $\mathcal{A}_{\mathcal{S}}$ is the mapping operator $\mathcal{A}$ constrained on $\mathcal{S}$, which is defined as $\mathcal{A}_{\mathcal{S}}\left(\boldsymbol{X}\right)=\left\{\boldsymbol{a}_{i}^{T}\boldsymbol{X}\boldsymbol{a}_{i}\right\}_{i\in\mathcal{S}}$. Denote the adjoint operator of $\mathcal{A}$ by $\mathcal{A}^*(\boldsymbol{\mu})=\sum_{i=1}^m \mu_i \ba_i\ba_i^T$, where $\mu_{i}$ is the $i$th entry of $\boldsymbol{\mu}$, $1\leq i\leq m$.}
 We use $\left\Vert \boldsymbol{X}\right\Vert $, $\left\Vert \boldsymbol{X}\right\Vert _{\mathrm{F}}$ and $\left\Vert \boldsymbol{X}\right\Vert_1$ to denote the spectral
norm, the Frobenius norm and the nuclear norm of the matrix $\boldsymbol{X}$, respectively, and use $\left\Vert\boldsymbol{x}\right\Vert_{p}$ to denote the $\ell_{p}$-norm of the vector $\boldsymbol{x}$. Let the singular value decomposition of the fixed rank-$r$ PSD matrix $\bX_{0}$ be $\bX_{0}=\boldsymbol{U}\boldsymbol{\Lambda}\boldsymbol{U}^{T}$, then the symmetric tangent space $T$ at $\bX_0$ is denoted by
\begin{equation*}
T:=\left\{ \boldsymbol{U}\bZ^T+\bZ\boldsymbol{U}^{T}\mid\bZ\in\mathbb{R}^{n\times r}\right\}.
\end{equation*}
We denote by $\mathcal{P}_{T}$ and $\mathcal{P}_{T^{\perp}}$ the orthogonal projection onto $T$ and its orthogonal complement, respectively. And for notational simplicity, we denote $\boldsymbol{H}_{T}:=\mathcal{P}_{T}\left(\boldsymbol{H}\right)$ and $\boldsymbol{H}_{T^{\perp}}:=\boldsymbol{H}-\mathcal{P}_{T}\left(\boldsymbol{H}\right)$ for any symmetric matrix $\boldsymbol{H}\in\mathbb{R}^{n\times n}$. Moreover, $\gamma$, $c$, $c_{1}$ and $c_{2}$ represent absolute constants, whose values may change according to context.

\subsection{Approximate Dual Certificate}\label{appro_dual_cert}
The following lemma suggests that under certain appropriate restricted isometry preserving properties of $\cA$, a properly constructed dual certificate can guarantee faithful recovery of the proposed algorithm \eqref{phaselift_outlier}.

\begin{lemma}[Approximate Dual Certificate for \eqref{phaselift_outlier}]\label{lemma-dual-certificate}
Denote a subset $\cS$ with $\frac{|\cS|}{m}: = \lceil  \frac{s_{0}}{13\sqrt{2r}} \rceil $, where $0<s_{0}<1$ is some constant, and the support of $\bbeta$ satisfies $\supp(\bbeta)\subseteq\cS$. Suppose that the mapping $\mathcal{A}$
obeys that for all symmetric matrices
$\boldsymbol{X}$,
\begin{equation}\label{isometryA}
\frac{1}{m}\left\Vert \mathcal{A}\left(\boldsymbol{X}\right)\right\Vert _{1}\leq\left(1+\frac{1}{10}\right)\left\|\boldsymbol{X}\right\|_1,
\end{equation}
and
\begin{equation}\label{isometryA_subset}
\frac{1}{|\mathcal{S}|}\left\Vert \mathcal{A}_{\mathcal{S}}\left(\boldsymbol{X}\right)\right\Vert _{1}\leq\left(1+\frac{1}{10}\right) \left\|\boldsymbol{X}\right\|_1,
\end{equation}
and for all matrices $\boldsymbol{X}\in T$,
\begin{equation}\label{lower-isometry}
\frac{1}{|\mathcal{S}^{\perp}|}\left\Vert \mathcal{A}_{\mathcal{S}^{\perp}}\left(\boldsymbol{X}\right)\right\Vert _{1}>\frac{1}{5}\left(1-\frac{1}{12}\right)\left\Vert \boldsymbol{X}\right\Vert _{\mathrm{F}},
\end{equation}
where $\mathcal{A}_{\mathcal{S}}$ and $\mathcal{A}_{\mathcal{S}^{\perp}}$ is the operator constrained on $\mathcal{S}$ and $\mathcal{S}^{\perp}$ respectively. Then if there exists a matrix $\boldsymbol{Y}=\mathcal{A}^*(\boldsymbol{\mu})$ that satisfies
\begin{equation}\label{certificate}
\boldsymbol{Y}_{T^{\perp}}\preceq - \frac{1}{r}\boldsymbol{I}_{T^{\perp}}, \quad \left\Vert \boldsymbol{Y}_{T}\right\Vert _{\mathrm{F}}\leq\frac{1}{13r} ,
\end{equation}
and
\begin{equation}\label{subgradient_lambda}
\left\{\begin{array}{cc}
\mu_i = \frac{9}{m}\sgn(\beta_i), &\quad  i\in\supp(\bbeta)\\
|\mu_i |\le \frac{9}{m}, & \quad i\notin \supp(\bbeta)
\end{array}\right.,
\end{equation}
the solution to \eqref{phaselift_outlier} satisfies
$$\left\|\hat{\bX}- \boldsymbol{X}_{0} \right\|_{\mathrm{F}}\leq c\frac{r \epsilon}{m},$$
where $c$ is a constant. 
\end{lemma}

\begin{proof}
Denote the solution to \eqref{phaselift_outlier} by $\hat{\boldsymbol{X}} = \boldsymbol{X}_{0}+\boldsymbol{H}\neq\boldsymbol{X}_{0}$, then we have $\hat{\boldsymbol{X}} \succeq 0$, $\boldsymbol{H}_{T^{\perp}}\succeq0$, and furthermore,
\begin{align*}
 \| \cA(\bH) - (\bbeta+\bw) \|_1 &= \|\bz -\cA(\bX_0 +\bH) \|_1 \\
 &= \|\bz -\cA(\hat{\bX}) \|_1 \\
 &\leq \|\bz -\cA(\bX_0) \|_1  = \| \bbeta+\bw\|_1,
\end{align*}
\yc{where the inequality follows from the optimality of $\hat{\bX}$ since both $\hat{\bX}$ and $\bX_0$ are feasible to \eqref{phaselift_outlier}.} Since
\begin{equation*}
  \| \cA(\bH) - (\bbeta+\bw) \|_1 =\| \cA_{\cS}(\bH) - \bbeta-\bw_{\cS} \|_1 + \|\cA_{\cS^{\perp}}(\bH)-\bw_{\cS^{\perp}} \|_1, 
\end{equation*}
and
\begin{equation*}
 \| \bbeta+\bw\|_1 = \|\bbeta+ \bw_{\cS}\|_1 + \|\bw_{\cS^{\perp}}\|_1, 
\end{equation*}
we have
\begin{align*}
\| \cA_{\cS^{\perp}}(\bH) \|_1 & \leq \| \cA_{\cS^{\perp}}(\bH)-\bw_{\cS^{\perp}} \|_1 +  \|\bw_{\cS^{\perp}}\|_1\\
& \leq  \| \bbeta+\bw\|_1  - \| \cA_{\cS}(\bH) - \bbeta-\bw_{\cS} \|_1+  \|\bw_{\cS^{\perp}}\|_1\\
& \leq  \|\bbeta+ \bw_{\cS}\|_1 -\| \cA_{\cS}(\bH) - \bbeta-\bw_{\cS} \|_1 +2 \|\bw_{\cS^{\perp}}\|_1 \\
& \leq \| \cA_{\cS}(\bH) \|_1+2 \|\bw_{\cS^{\perp}}\|_1,
\end{align*}
where the last inequality follows from the triangle inequality. We could further bound
\begin{align}
 \| \cA_{\cS^{\perp}}(\bH_{T}) \|_1  &\leq  \| \cA_{\cS^{\perp}}(\bH) \|_1 + \| \cA_{\cS^{\perp}}(\bH_{T^{\perp}}) \|_1 \nonumber \\
 & \leq \| \cA_{\cS}(\bH) \|_1+ \| \cA_{\cS^{\perp}}(\bH_{T^{\perp}}) \|_1 + 2 \|\bw_{\cS^{\perp}}\|_1 \nonumber \\
 &\leq\| \cA_{\cS}(\bH_T) \|_1+\| \cA_{\cS}(\bH_{T^{\perp}}) \|_1 \nonumber \\
& \quad + \| \cA_{\cS^{\perp}}(\bH_{T^{\perp}}) \|_1 + 2 \|\bw_{\cS^{\perp}}\|_1 \nonumber \\
& = \| \cA_{\cS}(\bH_T) \|_1+\| \cA(\bH_{T^{\perp}}) \|_1 + 2 \|\bw_{\cS^{\perp}}\|_1. \label{Hbound}
\end{align}

Our assumptions on $\mathcal{A}$ imply that
\begin{align*}
& \left(1+\frac{1}{10}\right)\mathrm{Tr}\left(\boldsymbol{H}_{T^{\perp}}\right) \\
& \geq\frac{1}{m}\left\Vert \mathcal{A}\left(\boldsymbol{H}_{T^{\perp}}\right)\right\Vert _{1} \\
& \geq \frac{1}{m}\left(  \| \cA_{\cS^{\perp}}(\bH_{T}) \|_1 -  \| \cA_{\cS}(\bH_T) \|_1- 2 \|\bw_{\cS^{\perp}}\|_1 \right) \\
& \geq \frac{|\mathcal{S}^{\perp}|}{5m}\left(1-\frac{1}{12}\right) \|\bH_T\|_{\mathrm{F}} -\frac{|\mathcal{S}|}{m}\left(1+\frac{1}{10}\right)\|\bH_T\|_{1} -\frac{2\epsilon}{m},
\end{align*}
where the first inequality follows from \eqref{isometryA} \yxl{due to $\left\Vert \boldsymbol{H}_{T^{\perp}} \right\Vert_{1} = \mathrm{Tr}\left( \boldsymbol{H}_{T^{\perp}} \right)$, as $\boldsymbol{H}_{T^{\perp}}\succeq0$}, the second inequality follows from \eqref{Hbound}, and the last inequality follows from \eqref{isometryA_subset} and \eqref{lower-isometry}. This gives
\begin{equation}\label{lower_traceHTperp}
\mathrm{Tr}\left(\boldsymbol{H}_{T^{\perp}}\right) \geq \left( \frac{|\mathcal{S}^{\perp}|}{6m}  - \frac{|\mathcal{S}|}{m} \sqrt{2r} \right)  \|\bH_T\|_{\mathrm{F}} - \frac{2\epsilon}{m},
\end{equation}
where we use the inequality $\|\bH_T\|_{\mathrm{1}}\le\sqrt{2r}\|\bH_T\|_{\mathrm{F}}$.

On the other hand, since $\boldsymbol{\mu}/(9/m)$ is a subgradient of the $\ell_1$-norm at $\bbeta$ from \eqref{subgradient_lambda}, we have
\begin{align*} 
\| \bbeta\|_1 + \left\langle \frac{m}{9}\boldsymbol{\mu} ,  \bw- \cA(\bH) \right\rangle & \leq \| \bw+ \bbeta - \cA(\bH)\|_1 \\
&\leq  \| \bbeta+\bw\|_1\leq  \| \bbeta\|_1+\| \bw\|_1,
\end{align*}
which, by a simple transformation, is 
\begin{align*} 
\left\langle  \boldsymbol{\mu} , \cA(\bH)\right\rangle& \geq \left\langle \boldsymbol{\mu}, \bw \right\rangle - \frac{9}{m}\| \bw\|_1 \\
&\geq  -\left(\|\boldsymbol{\mu}\|_{\infty}+\frac{9}{m} \right) \| \bw\|_1\geq -\frac{18\epsilon}{m}.
\end{align*}
Then with 
$$\langle \boldsymbol{H}, \boldsymbol{Y}\rangle =\langle \mathcal{A}(\boldsymbol{H}), \boldsymbol{\mu} \rangle,$$
we can get
\begin{align*}
 -\frac{18\epsilon}{m}& \leq \langle \mathcal{A}(\boldsymbol{H}), \boldsymbol{\mu} \rangle  = \langle \bH, \bY \rangle \\
 & = \langle \boldsymbol{H}_{T},\boldsymbol{Y}_T \rangle +\langle \boldsymbol{H}_{T^\perp}, \boldsymbol{Y}_{T^\perp} \rangle \\
& \leq \| \bY_T\|_{\mathrm{F}}  \left\Vert \boldsymbol{H}_{T}\right\Vert _{\mathrm{F}}  - \frac{1}{r} \langle \bH_{T^\perp}, \boldsymbol{I}_{T^\perp} \rangle  \\
& \leq  \frac{1}{13 r} \| \bH_T\|_{\mathrm{F}} - \frac{1}{r}  \mbox{Tr}(\boldsymbol{H}_{T^\perp}),
\end{align*}
which gives
\begin{equation} \label{upper_traceHTperp}
\mbox{Tr}(\boldsymbol{H}_{T^\perp}) \leq \frac{1}{13}\left\Vert \boldsymbol{H}_{T}\right\Vert _{\mathrm{F}} + \frac{18r\epsilon}{m} .
\end{equation}
Combining with \eqref{lower_traceHTperp}, we know
$$\left( \frac{|\mathcal{S}^{\perp}|}{6m}  - \frac{|\mathcal{S}|}{m} \sqrt{2r} \right)  \|\bH_T\|_{\mathrm{F}} - \frac{2\epsilon}{m} \leq  \frac{1}{13}\left\Vert \boldsymbol{H}_{T}\right\Vert _{\mathrm{F}} + \frac{18r\epsilon}{m} . $$

Since $ \frac{|\mathcal{S}^{\perp}|}{6m}  - \frac{|\mathcal{S}|}{m} \sqrt{2r} -\frac{1}{13}  >0$ under the assumption on $\frac{\left\vert\mathcal{S}\right\vert}{m}$ in Lemma~\ref{lemma-dual-certificate}, we have
\begin{equation*}
 \|\bH_T\|_{\mathrm{F}} \leq \frac{20r\epsilon}{m\left( \frac{|\mathcal{S}^{\perp}|}{6m}  - \frac{|\mathcal{S}|}{m} \sqrt{2r} -\frac{1}{13}\right) } \leq c_{1}\frac{ r \epsilon}{m},
\end{equation*}
where $c_1$ is some fixed constant. Finally, we have 
\begin{align*}
 \|\hat{\boldsymbol{X}} - \boldsymbol{X}_{0} \|_{\mathrm{F}}   &\leq \|\boldsymbol{H}_T \|_{\mathrm{F}}+ \| \boldsymbol{H}_{T^\perp}\|_{\mathrm{F}} \\
& \leq \|\boldsymbol{H}_T \|_{\mathrm{F}}+ \mbox{Tr}(\bH_{T^{\perp}}) \\
& \leq \left(1+\frac{1}{13}\right) \left\Vert \boldsymbol{H}_{T}\right\Vert _{\mathrm{F}}+\frac{18r\epsilon}{m} \leq c \frac{r \epsilon}{m},
\end{align*}
for some constant $c$.
\end{proof}

\subsection{Restricted Isometry of $\cA$} \label{sec:isometryA}

The first two conditions \eqref{isometryA} and \eqref{isometryA_subset} in Lemma~\ref{lemma-dual-certificate} are supplied straightforwardly in the following lemma as long as $m\geq c nr$ and $\left\vert \mathcal{S} \right\vert = c_{1}m/r \ge c_{2}n$ for some constants $c$, $c_{1}$ and $c_{2}$.
\begin{lemma}[\cite{candes2013phaselift}]\label{lemma-RIP1}
Fix any $\delta\in(0,\frac{1}{2})$
and assume $m\geq20\delta^{-2}n$. Then for all PSD matrices $\boldsymbol{X}$, one has
\[
\left(1-\delta\right) \left\|\boldsymbol{X}\right\|_1 \leq\frac{1}{m}\left\Vert \mathcal{A}\left(\boldsymbol{X}\right)\right\Vert _{1}\leq\left(1+\delta\right) \left\|\boldsymbol{X}\right\|_1
\]
with probability exceeding $1-2e^{-m\epsilon^{2}/2}$, where $\epsilon^{2}+\epsilon=\frac{\delta}{4}$. The right hand side holds for all symmetric matrices.
\end{lemma}

The third condition \eqref{lower-isometry} in Lemma~\ref{lemma-dual-certificate} can be obtained using the mixed-norm RIP-$\ell_{2}/\ell_{1}$ provided in \cite{chen2015exact} as long as $m\geq c nr$ and $\left\vert\mathcal{S}\right\vert\le c_{1}m$ for some constants $c$ and $c_{1}$.

\begin{lemma}[\cite{chen2015exact}]\label{lemma-isometry-manifold}
Suppose the sensing vectors $\boldsymbol{a}_{i}$'s are composed of i.i.d. sub-Gaussian entries, then there exist positive universal constants $c_{1}$, $c_{2}$, $c_{3}$ such that, provided that $m> c_{3}nr$, for all matrices $\boldsymbol{X}$ of rank at most $r$, one has
\begin{equation*}
 \left(1-\delta^{\mathrm{lb}}_{r}\right) \left\Vert\boldsymbol{X}\right\Vert_{\mathrm{F}} \le \frac{2}{m}\left\Vert\mathcal{B}\left(\boldsymbol{X}\right)\right\Vert_{1} \le \left(1+\delta^{\mathrm{ub}}_{r}\right) \left\Vert\boldsymbol{X}\right\Vert_{\mathrm{F}},
\end{equation*}
with probability exceeding $1-c_{1}e^{-c_{2}m}$, where $\delta^{\mathrm{lb}}_{r}$ and $\delta^{\mathrm{ub}}_{r}$ are defined as the RIP-$\ell_{2}/\ell_{1}$ constants. And the operator $\mathcal{B}$ represents the linear transformation that maps $\boldsymbol{X}\in\mathbb{R}^{n\times n}$ to $\left\{\mathcal{B}_{i}\left(\boldsymbol{X}\right)\right\}^{m/2}_{i=1}\in\mathbb{R}^{m/2}$, where $\mathcal{B}_{i}\left(\boldsymbol{X}\right):=\langle \boldsymbol{a}_{2i-1}\boldsymbol{a}^{T}_{2i-1} - \boldsymbol{a}_{2i}\boldsymbol{a}^{T}_{2i},\boldsymbol{X}\rangle$.
\end{lemma}

The third condition \eqref{lower-isometry} can be easily validated from the lower bound by setting $\delta^{\mathrm{lb}}_{r}$ appropriately, since $\frac{2}{m}\left\Vert\mathcal{B}\left(\boldsymbol{X}\right)\right\Vert_{1}  \le \frac{2}{m}\sum_{i=1}^{m/2} \left(\left\vert\langle \boldsymbol{a}_{2i-1}\boldsymbol{a}^{T}_{2i-1},\boldsymbol{X}\rangle\right\vert + \left\vert\langle \boldsymbol{a}_{2i}\boldsymbol{a}^{T}_{2i},\boldsymbol{X}\rangle\right\vert \right) = 2\| \mathcal{A}\left(\boldsymbol{X}\right)\|_1$.



\subsection{Construction of Dual Certificate} \label{dual-construction}

For notational simplicity, let $\alpha_{0}:=\mathbb{E}  Z^{2}{\bf 1}_{\left\{  \left\vert Z\right\vert \leq 3 \right\} } \approx 0.9707$, $\beta_{0}:=\mathbb{E}  Z^{4}{\bf 1}_{\left\{  \left\vert Z\right\vert \leq 3 \right\} } \approx 2.6728$ and $\theta_{0}:=\mathbb{E}  Z^{6}{\bf 1}_{\left\{  \left\vert Z\right\vert \leq 3 \right\} } \approx 11.2102$ for a standard Gaussian random variable $Z$, where ${\bf 1}_{\mathcal{E}}$ is an indicator function with respect to an event $\mathcal{E}$.

Consider that the singular value decomposition of  a PSD matrix $\boldsymbol{X}_0$ of rank at most $r$ can be represented as $\boldsymbol{X}_0=\sum_{i=1}^{r}\lambda_{i}\boldsymbol{u}_{i}\boldsymbol{u}_{i}^T$, then inspired by \cite{candes2012solving,hand2016phaselift}, we construct $\bY$ as
\begin{align*}
\boldsymbol{Y}
&:=\frac{1}{m}\sum_{j\in\mathcal{S}^{\perp}} \Big[ \frac{1}{r} \sum_{i=1}^r \left|\boldsymbol{a}_{j}^{T}\boldsymbol{u}_{i}\right|^{2} {\bf 1}_{\left\{  \left|\boldsymbol{a}_{j}^{T}\boldsymbol{u}_{i}\right| \leq 3 \right\} }\\
&\quad - (\alpha_{0}+\frac{\beta_{0}-\alpha_{0}}{r})  \Big] \cdot\boldsymbol{a}_{j}\boldsymbol{a}_{j}^{T}  + \frac{9}{m}\sum_{j\in\mathcal{S}} \chi_{j} \boldsymbol{a}_{j}\boldsymbol{a}_{j}^{T}  \\
& := \bY^{(0)}-\bY^{(1)}  + \bY^{(2)},
\end{align*}
where
\begin{align*}
\bY^{(0)} & =\frac{1}{m} \sum_{j\in\mathcal{S}^{\perp}} \left[\frac{1}{r} \sum_{i=1}^r \left|\boldsymbol{a}_{j}^{T}\boldsymbol{u}_{i}\right|^{2} {\bf 1}_{\left\{  \left|\boldsymbol{a}_{j}^{T}\boldsymbol{u}_{i}\right| \leq 3\right\} } \right]\boldsymbol{a}_{j}\boldsymbol{a}_{j}^{T}, \\
\bY^{(1)} & =  \frac{1}{m}\left(\alpha_{0}+\frac{\beta_{0}-\alpha_{0}}{r}\right)\sum_{j\in\mathcal{S}^{\perp}}   \boldsymbol{a}_{j}\boldsymbol{a}_{j}^{T}, \\
\bY^{(2)} &= \frac{9}{m}\sum_{j\in\mathcal{S}} \chi_{j}\boldsymbol{a}_{j}\boldsymbol{a}_{j}^{T}.
\end{align*}
We set $\chi_j =  \mathrm{sgn}\left(\beta_{j}\right) $ if $j\in\mbox{supp}(\bbeta)$, otherwise $\chi_j$'s are i.i.d. Rademacher random variables with $\mathbb{P}\left\{\chi_j=1\right\}=\mathbb{P}\left\{\chi_j=-1\right\}=1/2$.

The construction immediately indicates that $\bY$ satisfies \eqref{subgradient_lambda}. We will show that $\bY$ satisfies \eqref{certificate} with high proability. In what follows, we separate the constructed $\bY$ into two parts and consider the bounds on $\bY^{(0)}-\bY^{(1)}$ and $\bY^{(2)}$ respectively.

\subsubsection{Proof of $\bY_{T^{\perp}}+\frac{1}{r}\bI_{T^\perp} \preceq 0$} 
First, by standard results in random matrix theory \cite[Corollary 5.35]{Vershynin2012}, we have
\begin{equation*}
 \left\| \frac{m}{\left\vert\mathcal{S}^{\perp}\right\vert} \boldsymbol{Y}^{(1)} - \left(\alpha_{0} + \frac{\beta_{0}-\alpha_{0}}{r}\right) \boldsymbol{I} \right\| \leq \frac{\beta_{0}}{40r}, 
 \end{equation*}
 with probability at least $1-2e^{-\gamma \left\vert\mathcal{S}^{\perp}\right\vert/r^{2}}$ for some constant $\gamma$ provided $\left\vert\mathcal{S}^{\perp}\right\vert\geq cnr^2$ for some constant $c$. In particular, this gives
\begin{equation}\label{part1_part1}
 \left\| \frac{m}{\left\vert\mathcal{S}^{\perp}\right\vert}\boldsymbol{Y}_{T^{\perp}}^{(1)} -\left(\alpha_{0} + \frac{\beta_{0}-\alpha_{0}}{r}\right) \boldsymbol{I}_{T^{\perp}} \right\| \leq \frac{\beta_{0}}{40r}.
 \end{equation}

 Let $\boldsymbol{a}^{\prime}_{j}= (\bI -\bU\bU^T)\boldsymbol{a}_{j}$ be the projection of $\boldsymbol{a}_{j}$ onto the orthogonal complement of the column space of $\bU$, then we have
 \begin{equation*}
 \boldsymbol{Y}_{T^{\perp}}^{(0)} = \frac{1}{m} \sum_{j\in\mathcal{S}^{\perp}} \boldsymbol{\epsilon}_j \boldsymbol{\epsilon}_j^T,
 \end{equation*}
 where $\boldsymbol{\epsilon}_j =  \left(\frac{1}{r}\sum_{i=1}^r \left|\boldsymbol{a}_{j}^{T}\boldsymbol{u}_{i}\right|^{2} {\bf 1}_{\left\{  \left|\boldsymbol{a}_{j}^{T}\boldsymbol{u}_{i}\right| \leq 3 \right\} } \right)^{1/2} \ba_j^{\prime}$ are i.i.d. copies of a zero-mean, isotropic and sub-Gaussian random vector $\boldsymbol{\epsilon}$, which satisfies $\mathbb{E}[\boldsymbol{\epsilon}\boldsymbol{\epsilon}^T] = \alpha_{0} \boldsymbol{I}_{T^{\perp}}$. Following \cite[Theorem 5.39]{Vershynin2012}, we have
\begin{equation}\label{YT0} 
\left\| \frac{m}{\left\vert\mathcal{S}^{\perp}\right\vert}\boldsymbol{Y}_{T^{\perp}}^{(0)} -  \alpha_{0}  \boldsymbol{I}_{T^{\perp}} \right\|  \leq \frac{\alpha_{0}}{40r},
\end{equation} 
 with probability at least $1-2e^{-\gamma \left\vert\mathcal{S}^{\perp}\right\vert/r^{2}}$ for some constant $\gamma$ provided $\left\vert\mathcal{S}^{\perp}\right\vert\geq cnr^2$ for some constant $c$. As a result, if $m\ge cnr^2$ for some large constant $c$ and $\left\vert\mathcal{S}\right\vert \le c_{1}m$ for some constant $c_{1}$ small enough, with probability at least $1-e^{-\gamma m/r^{2}}$, there exists 
\begin{align}
&\left\Vert \boldsymbol{Y}_{T^{\perp}}^{(0)}  - \boldsymbol{Y}_{T^{\perp}}^{(1)}  +  \frac{\beta_{0}-\alpha_{0}}{r}\boldsymbol{I}_{T^{\perp}} \right\Vert  \nonumber \\
&\le \left\Vert \boldsymbol{Y}_{T^{\perp}}^{(0)}  - \boldsymbol{Y}_{T^{\perp}}^{(1)}  +  \frac{\beta_{0}-\alpha_{0}}{r}\frac{\left\vert\mathcal{S}^{\perp}\right\vert}{m}\boldsymbol{I}_{T^{\perp}} \right\Vert \nonumber\\
&\quad + \left(1-\frac{\left\vert\mathcal{S}^{\perp}\right\vert}{m}\right)\frac{\beta_{0}-\alpha_{0}}{r} \nonumber\\
&\le \left\Vert \frac{m}{\left\vert\mathcal{S}^{\perp}\right\vert}\boldsymbol{Y}_{T^{\perp}}^{(0)} - \frac{m}{\left\vert\mathcal{S}^{\perp}\right\vert}\boldsymbol{Y}_{T^{\perp}}^{(1)} +  \frac{\beta_{0}-\alpha_{0}}{r}\boldsymbol{I}_{T^{\perp}}\right\Vert + \frac{\left\vert\mathcal{S}\right\vert}{m} \frac{\beta_{0}-\alpha_{0}}{r} \nonumber\\
&\le \frac{\beta_{0}}{30r} + \frac{\alpha_{0}}{60r}.\label{part1_YT0_toge}
\end{align}

Next, let's check $\left\Vert\bY_{T^{\perp}}^{(2)}\right\Vert$. Since $\bY^{(2)} = \frac{1}{m}\sum_{j\in\cS}9\chi_j\ba_j\ba_j^T$, where $\mathbb{E}[9\chi_j\ba_j\ba_j^T]=\boldsymbol{0}$, by \cite[Theorem 5.39]{Vershynin2012} we have
\begin{equation*}
\left\|  \bY^{(2)} \right\| =\frac{|\cS|}{m}  \left\| \frac{1}{|\cS|}\sum_{j\in \cS}9 \chi_j  \ba_j\ba_j^T \right\|  \leq \frac{1}{10r},
\end{equation*}
with probability at least $1-2\exp(-\gamma m/r)$ as long as $m\geq cnr^2$ and $|\cS| = c_{1}m/r \ge c_{2}nr$, for some constants $c$, $c_{1}$ and $c_{2}$. In particular, this gives
\begin{equation}
\left\|  \bY_{T^\perp}^{(2)} \right\|  \leq \frac{1}{10r} .
\end{equation}

Putting this together with \eqref{part1_YT0_toge}, we can obtain that if $m\ge cnr^2$ and $|\cS| = c_{1}m/r \ge c_{2}nr$ for some constants $c$, $c_{1}$ and $c_{2}$, with probability at least $1-e^{-\gamma m/r^{2}}$, 
\begin{align*}
\left\Vert\bY_{T^\perp} + \frac{1.7}{r} \bI_{T^{\perp}}\right\Vert & = \left\Vert\bY_{T^\perp}^{(0)}-\bY_{T^\perp}^{(1)}  + \bY_{T^\perp}^{(2)}  + \frac{1.7}{r} \bI_{T^{\perp}}\right\Vert \\
&\le \left( \frac{\alpha_{0}}{60} +\frac{\beta_{0}}{30}  + 0.11\right) \frac{1}{r} \le \frac{0.25}{r}.
\end{align*}

\subsubsection{Proof of $\left\Vert\bY_{T} \right\Vert_{\mathrm{F}}\leq \frac{1}{13r}$} 

Let $\tilde{\bY}=\left(\bY^{(0)}-\bY^{(1)}\right)\bU$, and $\tilde{\bY}^{\prime} = \left(\bI-\bU \bU^{T}\right)\tilde{\bY}$ be the projection of $\tilde{\bY}$ onto the orthogonal complement of $\bU$, then we have 
\begin{equation}\label{YT0_F}
\left\Vert\bY^{(0)}_T-\bY^{(1)}_{T}\right\Vert_{\mathrm{F}}^2=\left\Vert \bU^{T}\tilde{\bY}\right\Vert_{\mathrm{F}}^{2}+ 2\left\Vert \tilde{\bY}^{\prime} \right\Vert_{\mathrm{F}}^2.
\end{equation}

First consider the term $\|\bU^{T}\tilde{\bY}\|_{\mathrm{F}}^{2}$ in \eqref{YT0_F}, where the $k$th column of $\bU^{T}\tilde{\bY}$ can be expressed explicitly as
\begin{equation*}
\begin{split}
&\left(\bU^{T}\tilde{\bY}\right)_{k}\\
&=\frac{1}{m}\sum_{j\in\mathcal{S}^{\perp}}\left[\frac{1}{r} \sum_{i=1}^r \left|\boldsymbol{a}_{j}^{T}\boldsymbol{u}_{i}\right|^{2} {\bf 1}_{\left\{  \left|\boldsymbol{a}_{j}^{T}\boldsymbol{u}_{i}\right| \leq 3 \right\} }- \left(\alpha_{0} + \frac{\beta_{0}-\alpha_{0}}{r}\right)   \right]  \\
&\quad \cdot \left(\boldsymbol{a}_{j}^{T}\boldsymbol{u}_{k}\right)\left(\bU^{T}\boldsymbol{a}_{j}\right)\\
&:=\frac{1}{m}\boldsymbol{\Phi}\boldsymbol{c}_k,
\end{split}
\end{equation*}
where $\boldsymbol{\Phi}\in\mathbb{R}^{r\times \left\vert\mathcal{S}^{\perp}\right\vert}$ is constructed by $\bU^{T}\boldsymbol{a}_{j}$'s, and $\boldsymbol{c}_{k}\in\mathbb{R}^{\left\vert\mathcal{S}^{\perp}\right\vert}$ is composed of $c_{k,j}$'s, each one expressed as
\begin{equation*}
\begin{split}
&c_{k,j} = \\
&  \left[\frac{1}{r}  \sum_{i=1}^r \left|\boldsymbol{a}_{j}^{T}\boldsymbol{u}_{i}\right|^{2} {\bf 1}_{\left\{ \left|\boldsymbol{a}_{j}^{T}\boldsymbol{u}_{i}\right| \leq 3 \right\} }- \left(\alpha_{0}+\frac{\beta_{0}-\alpha_{0}}{r}\right)   \right]   \left(\boldsymbol{a}_{j}^{T}\boldsymbol{u}_{k}\right),
\end{split}
\end{equation*}
with
\begin{equation*}
\begin{split}
 \mathbb{E}[c_{k,j}^2] &=\frac{1}{r^{2}}\left(\theta_{0} + \left(r-1\right)\beta_{0} - \beta^{2}_{0} - \left(r-1\right)\alpha^{2}_{0} \right)\\
 &= \frac{1}{r}\left(\beta_{0}-\alpha^{2}_{0}\right) + \frac{1}{r^{2}}\left(\theta_{0}+\alpha^{2}_{0}-\beta^{2}_{0}-\beta_{0}\right)  \leq \frac{4.07}{r}.
 \end{split}
 \end{equation*}
Note that $c^{2}_{k,j}$'s are i.i.d. sub-exponential random variables with $\left\Vert c^{2}_{k,j}\right\Vert_{\psi_{1}}\le K$, for some constant $K$, then according to \cite[Corollary 5.17]{Vershynin2012},  
\begin{equation*}
\mathbb{P}\left\{\left\vert\sum_{j\in\mathcal{S}^{\perp}}\left(c^{2}_{k,j}-\mathbb{E}c^{2}_{k,j}\right)\right\vert \ge \frac{\epsilon}{r}\left\vert\mathcal{S}^{\perp}\right\vert \right\} \le 2\mathrm{exp}\left(-c\frac{\epsilon^{2}\left\vert\mathcal{S}^{\perp}\right\vert}{K^{2}r^{2}}\right),
\end{equation*}
which shows that as long as $\left\vert\mathcal{S}\right\vert\le c_{1}m$, for some constants $c$ and $c_{1}$, 
\begin{equation*}
\left\Vert\boldsymbol{c}_{k}\right\Vert^{2}_{2}\le \frac{4.07+c}{r}m\le \frac{4.1m}{r}
\end{equation*}
holds with probability at least $1-e^{-\gamma m/r^{2}}$. Furthermore, for a fixed vector $\boldsymbol{x}\in\mathbb{R}^{\left\vert\mathcal{S}^{\perp}\right\vert}$ obeying $\left\Vert\boldsymbol{x}\right\Vert_{2}=1$, $\left\Vert\boldsymbol{\Phi}\boldsymbol{x}\right\Vert^{2}_{2}$ is distributed as a chi-square random variable with $r$ degrees of freedom. From \cite[Lemma 1]{laurent2000adaptive}, we have 
\begin{equation*}
\left\Vert\boldsymbol{\Phi}\boldsymbol{x}\right\Vert^{2}_{2} \le \frac{m}{12000r^{2}},
\end{equation*}
with probability at least $1-e^{-\gamma m/r^{2}}$, provided $m\ge cnr^{2}$ for some sufficiently large constant $c$. Therefore, we can obtain 
\begin{equation*}
\left\Vert \left(\bU^{T}\tilde{\bY}\right)_{k} \right\Vert^{2}_{2} = \frac{1}{m^2}\left\|\boldsymbol{\Phi}\frac{\boldsymbol{c}_k}{\|\boldsymbol{c}_k\|_2} \right\|_2^2 \|\boldsymbol{c}_k\|_2^2 \le \frac{1}{2700r^{3}},
\end{equation*}
which yields 
\begin{equation}\label{UT_tildeY}
\|\bU^{T}\tilde{\bY}\|_{\mathrm{F}}^{2}= \sum_{k=1}^{r}\left\Vert \left(\bU^{T}\tilde{\bY}\right)_{k} \right\Vert^{2}_{2} \le\frac{1}{2700r^{2}},
\end{equation} 
with probability at least $1-e^{-\gamma m/r^{2}}$, when $m\ge cnr^{2}$ and  $\left\vert\mathcal{S}\right\vert \le c_{1}m$.

To bound the second term in \eqref{YT0_F}, we could adopt the same techniques as before. The $k$th column of $\tilde{\bY}^{\prime}$ can be expressed explicitly as
\begin{equation*}
\begin{split}
&\left(\tilde{\bY}^{\prime}\right)_{k}\\
&=\frac{1}{m}\sum_{j\in\mathcal{S}^{\perp}}\left[ \frac{1}{r}\sum_{i=1}^r \left|\boldsymbol{a}_{j}^{T}\boldsymbol{u}_{i}\right|^{2} {\bf 1}_{\left\{ \left|\boldsymbol{a}_{j}^{T}\boldsymbol{u}_{i}\right| \leq 3 \right\} }- \left(\alpha_{0} + \frac{\beta_{0} - \alpha_{0}}{r}\right)  \right]\\
&\quad \cdot   (\boldsymbol{a}_{j}^{T}\boldsymbol{u}_{k}) \left(\boldsymbol{I} - \bU \bU^{T}\right)\boldsymbol{a}_{j}\\
& :=\frac{1}{m}\sum_{j\in\mathcal{S}^{\perp}} c_{k,j} \ba_{j}^{\prime} :=\frac{1}{m}\boldsymbol{\Psi}\boldsymbol{c}_k,
\end{split}
\end{equation*}
where $\boldsymbol{\Psi}\in\mathbb{R}^{n\times\left\vert\mathcal{S}^{\perp}\right\vert}$ is constructed by $\ba_{j}'$'s, each of which, as a reminder, is the projection of $\ba_{j}$ onto the orthogonal complement of the column space of $\bU$ as $\ba_{j}^{\prime}=\left(\boldsymbol{I} - \bU \bU^{T}\right)\ba_j$. Equivalently, $\boldsymbol{\Psi}= \left(\boldsymbol{I} - \bU \bU^{T}\right) \boldsymbol{A}$, where $\boldsymbol{A}\in\mathbb{R}^{n\times\left\vert\mathcal{S}^{\perp}\right\vert}$ is constructed by $\boldsymbol{a}_{j}$'s, $j\in\mathcal{S}^{\perp}$. For a fixed vector $\boldsymbol{x}\in\mathbb{R}^{\left\vert\mathcal{S}^{\perp}\right\vert}$ obeying $\left\Vert\boldsymbol{x}\right\Vert_{2}=1$, we have $\left\Vert\boldsymbol{\Psi}\boldsymbol{x}\right\Vert^{2}_{2} = \left\Vert\left(\boldsymbol{I} - \bU \bU^{T}\right) \boldsymbol{A}\boldsymbol{x}\right\Vert^{2}_{2}\le \left\Vert\boldsymbol{A}\boldsymbol{x}\right\Vert^{2}_{2}$, where $\left\Vert\boldsymbol{A}\boldsymbol{x}\right\Vert^{2}_{2}$ is distributed as a chi-square random variable with $n$ degrees of freedom. Again \cite[Lemma 1]{laurent2000adaptive} tells us  
\begin{equation*}
\left\Vert\boldsymbol{\Psi}\boldsymbol{x}\right\Vert^{2}_{2} \le \left\Vert\boldsymbol{A}\boldsymbol{x}\right\Vert^{2}_{2} \le \frac{m}{12000r^{2}},
\end{equation*}
with probability exceeding $1-e^{-\gamma m/r^{2}}$, provided $m\ge cnr^{2}$ for a sufficiently large constant $c$. Hence,
\begin{equation*}
\left\Vert \left(\tilde{\bY}^{\prime}\right)_{k}\right\Vert_2^2 \leq \frac{1}{m^2}\left\|\boldsymbol{\Psi}\frac{\boldsymbol{c}_k}{\|\boldsymbol{c}_k\|_2} \right\|_2^2 \|\boldsymbol{c}_k\|_2^2 \leq   \frac{1}{2700r^3},
\end{equation*}
which leads to 
\begin{equation}\label{YprimeF}
\left\|\tilde{\bY}^{\prime}\right\|_{\mathrm{F}}^2 =\sum_{k=1}^{r}\left\Vert \left(\tilde{\bY}^{\prime}\right)_{k}\right\Vert_2^2  \leq \frac{1}{2700r^2}.
\end{equation}
Then, combining \eqref{UT_tildeY} and \eqref{YprimeF}, we know that 
\begin{equation}
\left\Vert\bY^{(0)}_T-\bY^{(1)}_{T}\right\Vert_{\mathrm{F}}=\sqrt{\left\Vert \bU^{T}\tilde{\bY}\right\Vert_{\mathrm{F}}^{2}+ 2\left\Vert \tilde{\bY}^{\prime} \right\Vert_{\mathrm{F}}^2} \le \frac{1}{30r}.
\end{equation}

Next, let's check $\|\bY_{T}^{(2)}\|_{\mathrm{F}}^{2}$, which can be written as 
\begin{equation*}
\|\bY_{T}^{(2)}\|_{\mathrm{F}}^{2}= \|\bU^{T} \bar{\bY} \|_{\mathrm{F}}^{2} + 2\| \bar{\bY}^{\prime} \|_{\mathrm{F}}^2,
\end{equation*}
where $\bar{\bY}=\bY^{(2)}\bU$ and $\bar{\bY}^{\prime} = (\bI-\bU\bU^T)\bar{\bY}$. For the first term $\|\bU^{T} \bar{\bY} \|_{\mathrm{F}}^{2}$, the $k$th column of $\bU^{T} \bar{\bY}$ can be formulated explicitly as
\begin{equation*}
\left(\bU^{T} \bar{\bY}\right)_{k} = \frac{1}{m}\sum_{j\in\mathcal{S}} 9 \chi_{j} \left(\boldsymbol{a}_{j}^{T}\boldsymbol{u}_{k}\right)\left(\bU^{T}\boldsymbol{a}_{j}\right):=\frac{1}{m}\bar{\boldsymbol{\Phi}}\boldsymbol{d}_k,
\end{equation*}
where $\bar{\boldsymbol{\Phi}}\in\mathbb{R}^{r\times \left\vert\mathcal{S}\right\vert}$ is constructed by $\bU^{T}\boldsymbol{a}_{j}$'s, and $\boldsymbol{d}_{k}\in\mathbb{R}^{\left\vert\mathcal{S}\right\vert}$ is composed of $d_{k,j}$'s, each one expressed as
\begin{equation*}
d_{k,j} = 9 \chi_{j} \left(\boldsymbol{a}_{j}^{T}\boldsymbol{u}_{k}\right),
\end{equation*}
with $\mathbb{E}[d_{k,j}^2] =81$. Note that $d^{2}_{k,j}$'s are i.i.d. sub-exponential random variables with $\left\Vert d^{2}_{k,j}\right\Vert_{\psi_{1}}\le K$, for some constant $K$, then based on \cite[Corollary 5.17]{Vershynin2012},  
\begin{equation*}
\mathbb{P}\left\{\left\vert\sum_{j\in\mathcal{S}}\left(d^{2}_{k,j}-\mathbb{E}d^{2}_{k,j}\right)\right\vert \ge \epsilon\left\vert\mathcal{S}\right\vert \right\} \le 2\mathrm{exp}\left(-c_{1}\frac{\epsilon^{2}\left\vert\mathcal{S}\right\vert}{K^{2}}\right),
\end{equation*}
which indicates that if $\left\vert\mathcal{S}\right\vert = cm/r$, for some constant $c$, 
\begin{equation*}
\left\Vert\boldsymbol{d}_{k}\right\Vert^{2}_{2}\le \left(81+c_{1}\right) \left\vert\mathcal{S}\right\vert \le 82 \left\vert\mathcal{S}\right\vert :=\delta_{0} \left\vert\mathcal{S}\right\vert
\end{equation*}
holds with probability at least $1-e^{-\gamma m/r}$. And for a fixed vector $\boldsymbol{x}\in\mathbb{R}^{\left\vert\mathcal{S}\right\vert}$ obeying $\left\Vert\boldsymbol{x}\right\Vert_{2}=1$, $\left\Vert\bar{\boldsymbol{\Phi}}\boldsymbol{x}\right\Vert^{2}_{2}$ is also a chi-square random variable with $r$ degrees of freedom, so   
\begin{equation*}
\left\Vert\bar{\boldsymbol{\Phi}}\boldsymbol{x}\right\Vert^{2}_{2} \le \frac{m}{2700\delta_{0}cr^{2}},
\end{equation*}
with probability at least $1-e^{-\gamma m/r^{2}}$, provided $m\ge c_{1}nr^{2}$ for some sufficiently large constant $c_{1}$. Thus we have 
\begin{equation*}
\left\Vert \left(\bU^{T}\bar{\bY}\right)_{k} \right\Vert^{2}_{2} = \frac{1}{m^2}\left\|\bar{\boldsymbol{\Phi}}\frac{\boldsymbol{d}_k}{\|\boldsymbol{d}_k\|_2} \right\|_2^2 \|\boldsymbol{d}_k\|_2^2 \le \frac{1}{2700r^{3}},
\end{equation*}
which gives 
\begin{equation}\label{UT_barY}
\|\bU^{T}\bar{\bY}\|_{\mathrm{F}}^{2}=\sum_{k=1}^{r}\left\Vert \left(\bU^{T}\bar{\bY}\right)_{k} \right\Vert^{2}_{2} \le\frac{1}{2700r^{2}},
\end{equation} 
with probability at least $1-e^{-\gamma m/r^{2}}$, when $m\ge c_{1}nr^{2}$ and  $\left\vert\mathcal{S}\right\vert = cm/r$, for some appropriate constants $c$ and $c_{1}$.

Now consider the second term $\| \bar{\bY}^{\prime} \|_{\mathrm{F}}^2$ in $\|\bY_{T}^{(2)}\|_{\mathrm{F}}^{2}$, where the $k$th column of $\bar{\bY}^{\prime}$ can be expressed explicitly as
\begin{equation*}
\begin{split}
\left(\bar{\bY}^{\prime}\right)_{k} &=\frac{1}{m}\sum_{j\in\mathcal{S}} 9 \chi_{j}  (\boldsymbol{a}_{j}^{T}\boldsymbol{u}_{k}) \left(\boldsymbol{I} - \bU \bU^{T}\right)\boldsymbol{a}_{j}\\
& :=\frac{1}{m}\sum_{j\in\mathcal{S}} d_{k,j} \ba_{j}^{\prime} :=\frac{1}{m}\bar{\boldsymbol{\Psi}}\boldsymbol{d}_k,
\end{split}
\end{equation*}
where $\bar{\boldsymbol{\Psi}}\in\mathbb{R}^{n\times\left\vert\mathcal{S}\right\vert}$ is constructed by $\ba_{j}'$'s. Also, we can decompose $\bar{\boldsymbol{\Psi}}$ as $\bar{\boldsymbol{\Psi}}= \left(\boldsymbol{I} - \bU \bU^{T}\right) \bar{\boldsymbol{A}}$, where $\bar{\boldsymbol{A}}\in\mathbb{R}^{n\times\left\vert\mathcal{S}\right\vert}$ is constructed by $\boldsymbol{a}_{j}$'s, $j\in\mathcal{S}$. For a fixed vector $\boldsymbol{x}\in\mathbb{R}^{\left\vert\mathcal{S}\right\vert}$ obeying $\left\Vert\boldsymbol{x}\right\Vert_{2}=1$, we have $\left\Vert\bar{\boldsymbol{\Psi}}\boldsymbol{x}\right\Vert^{2}_{2} = \left\Vert\left(\boldsymbol{I} - \bU \bU^{T}\right) \bar{\boldsymbol{A}}\boldsymbol{x}\right\Vert^{2}_{2}\le \left\Vert\bar{\boldsymbol{A}}\boldsymbol{x}\right\Vert^{2}_{2}$, where $\left\Vert\bar{\boldsymbol{A}}\boldsymbol{x}\right\Vert^{2}_{2}$ is a chi-square random variable with $n$ degrees of freedom as well. Since we already know that provided $m\ge c_{1}nr^{2}$ for a sufficiently large constant $c_{1}$,  
\begin{equation*}
\left\Vert\bar{\boldsymbol{\Psi}}\boldsymbol{x}\right\Vert^{2}_{2} \le \left\Vert\bar{\boldsymbol{A}}\boldsymbol{x}\right\Vert^{2}_{2} \le \frac{m}{2700\delta_{0}cr^{2}},
\end{equation*}
with probability exceeding $1-e^{-\gamma m/r^{2}}$, we can have 
\begin{equation*}
\left\Vert \left(\bar{\bY}^{\prime}\right)_{k}\right\Vert_2^2 \leq \frac{1}{m^2}\left\|\bar{\boldsymbol{\Psi}}\frac{\boldsymbol{d}_k}{\|\boldsymbol{d}_k\|_2} \right\|_2^2 \|\boldsymbol{d}_k\|_2^2 \leq   \frac{1}{2700r^3},
\end{equation*}
and a further result
\begin{equation}\label{YprimeF_bar}
\left\|\bar{\bY}^{\prime}\right\|_{\mathrm{F}}^2 =\sum_{k=1}^{r}\left\Vert \left(\bar{\bY}^{\prime}\right)_{k}\right\Vert_2^2  \leq \frac{1}{2700r^2},
\end{equation}
which, combining with \eqref{UT_barY}, leads to
\begin{equation}
\left\Vert\bY^{(2)}_T\right\Vert_{\mathrm{F}}=\sqrt{\left\Vert \bU^{T}\bar{\bY}\right\Vert_{\mathrm{F}}^{2}+ 2\left\Vert \bar{\bY}^{\prime} \right\Vert_{\mathrm{F}}^2} \le \frac{1}{30r}.
\end{equation}

Finally we can obtain that if $m\ge cnr^2$ and $\left\vert\mathcal{S}\right\vert = c_{1}m/r$, for some constants $c$ and  $c_{1}$, with probability at least $1-e^{-\gamma m/r^{2}}$,  
\begin{equation}
\left\Vert\bY_{T} \right\Vert_{\mathrm{F}} = \left\Vert \bY^{(0)}_T-\bY^{(1)}_{T} + \bY^{(2)}_T  \right\Vert_{\mathrm{F}} \leq \frac{1}{15r}.
\end{equation}

\subsection{Proof of Theorem~\ref{main}} \label{main_theorem_proof}
The required restricted isometry properties of the linear mapping $\mathcal{A}$ are supplied in Section~\ref{sec:isometryA} and a valid appropriate dual certificate is constructed in Section~\ref{dual-construction}, therefore, Theorem~\ref{main} can be straightforwardly obtained from the Lemma~\ref{lemma-dual-certificate} in Section~\ref{appro_dual_cert}.

\section{Conclusion} \label{sec:conclusion}

In this paper, we address the problem of estimating a low-rank PSD matrix $\bX\in\mathbb{R}^{n\times n}$ from rank-one measurements that are possibly corrupted by arbitrary outliers and bounded noise. This problem has many applications in covariance sketching, phase space tomography, and noncoherent detection in communications. It is shown that with an order of $nr^2$ random Gaussian sensing vectors, a PSD matrix of rank-$r$ can be robustly recovered by minimizing the $\ell_1$-norm of the observation residual within the semidefinite cone with high probability, even when a fraction of the measurements are adversarially corrupted. This convex formulation eliminates the need for trace minimization and tuning of parameters without prior knowledge of the outliers. Moreover, a non-convex subgradient descent algorithm is proposed with excellent empirical performance, when additional information of the rank of the PSD matrix is available. For future work, it would be interesting to theoretically justify the proposed non-convex algorithm. Finally, we note that very recently one of the authors proposed a median-truncated gradient descent algorithm for phase retrieval under a constant proportion of outliers with provable performance guarantees in \cite{zhang2016provable}, which might be possible to extend to the problem of robust low-rank PSD matrix recovery considered in this paper and will be pursued elsewhere.

\yc{\section*{Acknowledgement}
We thank the anonymous reviewers for their valuable suggestions that greatly improved the quality of this paper.}

\bibliographystyle{IEEEtran}
\bibliography{bibfileToeplitzPR,bibfileToeplitzPR2}

\end{document}